\documentclass{article}

\usepackage[final]{graphicx}
\usepackage{subcaption}

\usepackage[utf8]{inputenc} 
\usepackage[T1]{fontenc}    
\usepackage{hyperref}       
\usepackage{url}            
\usepackage{booktabs}       
\usepackage{amsfonts}       
\usepackage{nicefrac}       
\usepackage{microtype}      
\usepackage{xcolor}         
\usepackage[english]{babel}
\usepackage{authblk,url}
\usepackage{amssymb,amsmath,amsthm,twoopt,xargs,mathtools,dsfont}
\usepackage{times,ifthen}
\usepackage{fancyhdr,xcolor}
\usepackage{authblk}
\usepackage{fancyhdr}

\usepackage[left=3cm,right=3cm,top=3cm,bottom=3cm]{geometry}

\title{Amortized backward variational inference in nonlinear state-space models}
\date{}

\author[$\dag$]{ Mathis Chagneux}
\author[$\ddag$]{\'Elisabeth Gassiat}
\author[$\star$]{ Pierre Gloaguen}
\author[$\top$]{Sylvain Le Corff}

\affil[$\dag$]{{\small T\'el\'ecom Paris, Institut Polytechnique de Paris, Palaiseau.}}
\affil[$\ddag$]{{\small Universit\'e Paris-Saclay, CNRS, Laboratoire de math\'ematiques d'Orsay.}}
\affil[$\star$]{{\small AgroParisTech,  UMR MIA 518, Palaiseau.}}
\affil[$\top$]{{\small Samovar, T\'el\'ecom SudParis, d\'epartement CITI, TIPIC, Institut Polytechnique de Paris, Palaiseau.}}

\newcommand{\uksymbol}{\ell}
\newcommandx{\bkmod}[2][1=]{ 
\ifthenelse{\equal{#1}{}}
{\kernel{B}_{#2}^\precpar}
{\kernel{B}_{#2}^\precpar}
}

\newcommand{\parvar}{\lambda}
\newcommand{\parvarspace}{\Lambda}

\newcommand{\precpar}{\varphi}
\newcommand{\intvect}[2]{\{ #1, #2 \}}
\newcommandx\tstatmod[2][1=]{
\ifthenelse{\equal{#1}{}}
	{\tstatletter^{\precpar}_{#2}}
	{\tau^{\precpar}_{#2}^{#1}}
}

\newcommand{\udlow}{\sigma_-}
\newcommand{\udup}{\sigma_+}

\newcommandx{\bkw}[2][1=]{ 
\ifthenelse{\equal{#1}{}}
{\kernel{B}_{#2}}
{\kernel{B}_{#2}}
}

\newcommand{\ud}[1]{\uksymbol_{#1}} 
\newcommand{\nset}{\mathbb{N}}

\newcommand{\1}{\mathbbm{1}} 
\newcommandx{\postmod}[2][1=]{
\ifthenelse{\equal{#1}{}}
	{\phi_{#2}^\parvar}
	{\phi_{#2}^\N}
}
\newcommand{\retrokmod}{\boldsymbol{\mathcal{L}}^{\parvar,\theta}}
\newcommand{\uk}[2]{\mathbf{L}_{#1}^{#2}}

\newcommand{\md}[2]{g_{#1}^{#2}}

\newcommand{\llh}[2]{\mathsf{L}_{#1}^{#2}}

\newcommandx\filtderiv[2][1=]{
\ifthenelse{\equal{#1}{}}
	{\eta_{#2}}
	{\eta_{#2}^\N}
}

\newcommand{\parvec}{\theta}
\newcommand{\parspace}{\Theta}
\newcommand{\tstatletter}{\kernel{T}}
\newcommandx\tstat[2][1=]{
\ifthenelse{\equal{#1}{}}
	{\tstatletter_{#2}}
	{\tau_{#2}^{#1}}
}
\newcommandx\tstathat[2][1=]{
\ifthenelse{\equal{#1}{}}
	{\tstatletter_{#2}}
	{\widehat{\tau}_{#2}^{#1}}
}
\newcommand{\af}[1]{h_{#1}}
\newcommand{\deriv}{\nabla_{\parvec}}

\newcommand{\kernel}[1]{\mathbf{#1}}

\newcommandx{\bk}[2][1=]{
\ifthenelse{\equal{#1}{}}
{\overleftarrow{\kernel{Q}}_{#2}}
{\overleftarrow{\kernel{Q}}_{#2}^{#1}}
}

\newcommandx{\bkhat}[2][1=]{
\ifthenelse{\equal{#1}{}}
{\widehat{\kernel{Q}}_{#2}}
{\widehat{\kernel{Q}}_{#2}^{#1}}
}

\newcommand{\hd}[2]{m_{#1}^{#2}}

\newcommandx{\addf}[2][1=]{
\ifthenelse{\equal{#1}{}}{\termletter_{#2}}{\bar{h}_{#2 | #1}}
}

\newcommand{\termletter}{\tilde{h}}
\newcommand{\N}{N}

\newcommandx{\K}[1][1=]{
\ifthenelse{\equal{#1}{}}{{\kletter}}{{\widetilde{\N}^{#1}}}}

\newcommand{\kletter}{\widetilde{\N}}

\def\1{\mathds{1}}
\def\pE{\mathbb{E}}

\newcommand{\esssup}[2][]
{\ifthenelse{\equal{#1}{}}{\left\| #2 \right\|_\infty}{\left\| #2 \right\|^2_{\infty}}}

\newcommand{\rset}{\ensuremath{\mathbb{R}}}

\newcommand{\kiss}[3][]
{\ifthenelse{\equal{#1}{}}{r_{#2|#3}}
{\ifthenelse{\equal{#1}{fully}}{r^{\star}_{#2|#3}}
{\ifthenelse{\equal{#1}{smooth}}{\tilde{r}_{#2|#3}}{\mathrm{erreur}}}}}

\newcommand{\chunk}[4][]%
{\ifthenelse{\equal{#1}{}}{\ensuremath{{#2}_{#3:#4}}}{\ensuremath{#2^#1}_{#3:#4}}
}

\newcommand{\kissforward}[3][]
{\ifthenelse{\equal{#1}{}}{p_{#2}}
{\ifthenelse{\equal{#1}{fully}}{p^{\star}_{#2}}
{\ifthenelse{\equal{#1}{smooth}}{\tilde{r}_{#2}}{\mathrm{erreur}}}}}

\newcommandx\post[2][1=]{
\ifthenelse{\equal{#1}{}}
	{\phi_{#2}}
	{\phi_{#2}^\N}
}

\newcommandx\posthat[2][1=]{
\ifthenelse{\equal{#1}{}}
	{\widehat{\phi}_{#2}}
	{\widehat{\phi}_{#2}^\N}
}

\newcommand{\adjfunc}[4][]
{\ifthenelse{\equal{#1}{}}{\ifthenelse{\equal{#4}{}}{\vartheta_{#2|#3}}{\vartheta_{#2|#3}(#4)}}
{\ifthenelse{\equal{#1}{smooth}}{\ifthenelse{\equal{#4}{}}{\tilde{\vartheta}_{#2|#3}}{\tilde{\vartheta}_{#2|#3}(#4)}}
{\ifthenelse{\equal{#1}{fully}}{\ifthenelse{\equal{#4}{}}{\vartheta^\star_{#2|#3}}{\vartheta^\star_{#2|#3}(#4)}}{\mathrm{erreur}}}}}

\newcommand{\XinitIS}[2][]
{\ifthenelse{\equal{#1}{}}{\ensuremath{\rho_{#2}}}{\ensuremath{\check{\rho}_{#2}}}}

\newcommand{\rmd}{\ensuremath{\mathrm{d}}}
\newcommand{\eqdef}{\ensuremath{:=}}
\newcommand{\eqsp}{\;}

\newcommand{\filt}[2][]%
{%
\ifthenelse{\equal{#1}{}}{\ensuremath{\phi_{#2}}}{\ensuremath{\phi_{#1,#2}}}%
}

\newcommand{\sumwght}[2][]{%
\ifthenelse{\equal{#1}{}}{\ensuremath{\Omega_{#2}}}{\ensuremath{\Omega_{#2}^{(#1)}}}}

\newcommand{\sumwghthat}[2][]{%
\ifthenelse{\equal{#1}{}}{\ensuremath{\widehat{\Omega}_{#2}}}{\ensuremath{\widehat{\Omega}_{#2}^{(#1)}}}}

\newcommand{\qg}[2]{\ell_{#1}^{#2}}

\newcommand{\backward}[1]{\overleftarrow{#1}}

\newcommand{\vd}[1]{q_{#1}^{\parvar}} 
\newcommand{\kvd}[1]{q_{#1}^{\parvar}} 
\newcommand{\ivd}[1]{\tilde{q}_{#1}^{\parvar}} 
\newcommand{\jvd}{q} 

\newcounter{example}[section]

\newcounter{hypH}
\newenvironment{hypH}{\refstepcounter{hypH}\begin{itemize}
\item[{\bf H\arabic{hypH}}]}{\end{itemize}}

\newtheorem{theorem}{Theorem}[section]
\newtheorem{lemma}[theorem]{Lemma}
\newtheorem{proposition}[theorem]{Proposition}

\begin{document}

\maketitle

\begin{abstract}

We consider the problem of state estimation in general state-space models using variational inference. For a generic variational family defined using  the same  backward decomposition  as the actual joint smoothing distribution, we establish for the first time that, under mixing assumptions, the variational approximation of expectations of additive state functionals induces an error which grows at most linearly in the number of observations. 
This guarantee is consistent with the known upper bounds for the approximation of smoothing distributions using standard Monte Carlo methods. Moreover, we propose an amortized inference framework where a neural network shared over all times steps outputs the parameters of the variational kernels.
We also study empirically parametrizations which allow analytical marginalization of the variational distributions, and therefore lead to efficient smoothing algorithms. 
Significant improvements are made over state-of-the art variational solutions, especially when the generative model depends on  a strongly nonlinear and noninjective mixing function.

\end{abstract}

\section{Introduction}

When generative data models involve so-called hidden or \textit{latent} states, providing statistical estimates of the latter given observed data - also known as \textit{state inference} - is the cornerstone of many machine learning algorithms \cite{dempster1977maximum, Kingma2014AutoEncodingVB}. 
Traditional models usually introduce low-dimensional states having directly interpretable meaning, while benefiting from accurate inference via exact or consistent Monte Carlo methods. 
In contrast, modern latent-data machine learning models are rooted in the so-called manifold hypothesis which views high dimensional data as originating from hidden representations in an unknown space and via a complex nonlinear mapping. In the context of unsupervised representation learning, state inference is a goal in itself. 
Due to the intricacy and dimentionality of the inverse problems involved, most of these works resort to a combination of deep neural networks (DNNs) and variational approximations which allow tractable inference and serve as a principled proxy for maximum likelihood estimation (MLE) \cite{Higgins2017betaVAELB, Locatello2020WeaklySupervisedDW}.

The particular case of dependent data is of special importance as it guarantees identifiability results \cite{Khemakhem2020VariationalAA}, especially in the \textit{sequential} setting \cite{gassiat2020, Hlv2021DisentanglingIF}. 
This in turn renews interest in a more solid theoretical understanding of the behaviour of sequential variational methods.
In this work, we focus on the case where the true generative model is assumed to be a \textit{state-space model} (SSM). 
In the general SSM litterature, theoretical analysis of the conditional distribution of the states given the observations - commonly referred to as the \textit{smoothing} distribution - has been extensively conducted to derive efficient estimation algorithms with good convergence properties. Among these works, a keystone in sequential inference is the computation of expected values of additive state functionals under the smoothing distribution, known as additive smoothing (\cite{cappe2005inference}, Chap. 4), and more precisely the control of the additive smoothing error when the target expectations are approximated.  Theoretical guarantees have been provided when the approximation is performed using a surrogate of the true smoothing distribution provided by Sequential Monte Carlo (SMC) methods  \cite{douc2011sequential, dubarry2013non, olsson2017efficient, gloaguen2022pseudo}. 
In addition, in \cite{gloaguen2022pseudo}, a control has also been derived when the smoothed expectations are computed under a biased joint 
distribution of the hidden states and the observations.

In parallel to these works, sequential variational methods rely on a tractable approximation of the smoothing distribution to compute these expectations. However, this variational approximation has to  account for the  dependencies implied by the data model \cite{Bayer2021MindTG},  and typically does not recover the true distribution in the limit of infinite data when using mean-field variational families. This is why introducing dependency in the variational family has been recently explored in the literature. In \cite{johnson2016}, the authors obtained promising results by combining conjugate graphical models with 
variational inference, see also \cite{Lin2018VariationalMP} for variational methods based on graphical models in the inference network fostering fast amortized inference. 
In \cite{krishnan2017structured}, the
variational approximation uses a forward decomposition, parameterized by recurrent neural
networks,  which allows to mimic the forward decomposition of the true posterior distribution. More recently, \cite{campbell2021online} proposed a variational family using the so-called \textit{backward factorization}. Such a choice has very appealing properties as it is prone to  online state estimation and parameter learning in SSMs.

However, the question of whether these variational families suited to SSMs lead to good variational approximations for additive smoothing remains open. 
Indeed, to the best of our knowledge, there are no theoretical results providing upper bounds on the state estimation error when using any (mean field or involving dependencies) variational posterior in state-space models. 
In this paper, we establish the first theoretical guarantees for the variational approximation of additive smoothing in state-space-models, see Proposition~\ref{prop:bias}.

In Section \ref{sec:theoretical_result}, we prove that, in the case of strongly mixing state hidden Markov models, the variational estimation error of smoothed additive functional grows at most linearly with the number of observations. In Section \ref{sec:algorithm}, we build a backward variational inference algorithm involving fully amortized networks and amenable to recursive learning.  
In Section \ref{sec:experiments}, we illustrate the theoretical results numerically, and additionally show that a linear Gaussian parametrization of the backward variational kernels can achieve good performance at a small computational cost, even in the case of a strongly nonlinear and noninjective observation model.

\section{Background}
\paragraph{Notations. } Let $\Theta\subset \rset^q$ be a parameter space and consider a  \textit{state-space model} depending on $\parvec\in\parspace$ where the hidden Markov chain  in $\rset^d$ is denoted by $(X_k)_{k\geqslant 0}$. 
The distribution of $X_0$ has density $\chi^\parvec$ with respect to the Lebesgue measure $\mu$ and for all $k \geqslant 0$, the conditional distribution of $X_{k+1} $ given $X_{0:k}$ has density $\hd{k}{\parvec}(X_{k},\cdot)$, where $a_{u:v}$ is a short-hand notation for $(a_u,\ldots,a_v)$ for $0\leqslant u \leqslant v$ and any sequence $(a_\ell)_{\ell\geqslant 0}$. 
In SSMs, it is assumed that this state is partially observed through an observation process $(Y_k)_{0\leqslant k \leqslant n}$ taking values in $\rset^m$. The observations $Y_{0:n}$ are assumed to be independent conditionally on $X_{0:n}$ and, for all $0\leqslant k \leqslant n$, the distribution of $Y_k$ given $X_{0:n}$ depends on $X_k$ only and has density $\md{k}{\parvec}(X_k,\cdot)$ with respect to the Lebesgue measure. 

In the following, for any measure $\nu$ on a measurable space $(\mathsf{X},\mathcal{X})$ and any measurable function $h$ on $\mathsf{X}$, write $\nu h = \int h(x)\nu(\rmd x)$. In addition, for any measurable spaces $(\mathsf{X},\mathcal{X})$ and $(\mathsf{Y},\mathcal{Y})$, any measure $\nu$ on  $(\mathsf{X},\mathcal{X})$, any kernel $K:(\mathsf{X},\mathcal{Y})\to \rset_+$ and any measurable function $h$ on $\mathsf{X}\times \mathsf{Y}$, write $K h: x\mapsto  \int h(x,y)K(x,\rmd y)$ and $\nu K h = \int h(x,y)\nu(\rmd x)K(x,\rmd y)$. For simplicity, if for all $x\in \mathsf{X}$, $K(x,\cdot)$ has a density $k(x,\cdot)$ with respect to a reference measure $\nu$, we write $k h: x\mapsto  \int h(x,y)K(x,\rmd y) = \int h(x,y)k(x,y)\nu(\rmd y)$. Let also $\1$ be the constant function which equals 1 on $\mathbb{R}^d$.

\subsection{Latent data models and additive state functionals}

In this context, for any $0\leqslant k_1 \leqslant k_2 \leqslant n$ the \textit{joint smoothing distribution} $\post{k_1:k_2}^\parvec$ is the conditional law of $X_{k_1:k_2}$ given $Y_{0:n}$. 
For any function $h$ from $\rset^{d \times (n + 1)}$ to $\rset^{\mathsf{d}}$, we define its \textit{smoothed expectation} when the model is parameterized by $\theta$ as:
\begin{align}
\label{eq:smoothing:expectation}
\post{0:n}^{\parvec} h &= \pE^{\parvec}\left[h\left(X_{0:n}\right)\vert Y_{0:n}\right] \\
&= \llh{n}{\parvec}(Y_{0:n})^{-1} \int h(x_{0:n}) \chi^\parvec(x_0)\md{0}{\parvec}(x_{0},Y_{0})\prod_{k=0}^{n-1}\qg{k}{\parvec}(x_{k},x_{k+1})\mu(\rmd x_{0:n})\eqsp,\nonumber
\end{align}
where\footnote{Note that the dependence of $\qg{k}{\parvec}$ on $Y_{k+1}$ is omitted in the notation for better clarity.} 
$$
\qg{k}{\parvec}(x_{k},x_{k+1}) = \hd{k}{\parvec}(x_{k}, x_{k+1})\md{k+1}{\parvec}(x_{k+1},Y_{k+1})
$$
and $\llh{n}{\parvec}(Y_{0:n})$ is the likelihood of the observations:
\begin{equation}
\label{eq:likelihood}
\llh{n}{\parvec}(Y_{0:n})  = \int \chi^\parvec(x_0)\md{0}{\parvec}(x_{0},Y_{0})\prod_{k=0}^{n-1} \qg{k}{\parvec}(x_{k},x_{k+1})\mu(\rmd x_{0:n})\eqsp.
\end{equation}

In the context of state-space models, \textit{additive state functionals} are functions $\af{0:n}$ from $\rset^{d \times (n + 1)}$ to $\rset^{\mathsf{d}}$ satisfying:
\begin{equation}
\label{eq:additive:functional}
\af{0:n}: x_{0:n} \mapsto \sum_{k=0}^{n-1}\addf{k}(x_{k},x_{k+1})\eqsp,
\end{equation}
where $\addf{k}:\rset^{d} \times \rset^{d}\to\rset^{\mathsf{d}}$. In Bayesian inference, point estimates of most quantities of interest are naturally expressed as posterior means of random functionals belonging to this class. For state inference at a fixed $\theta$, i.e. the recovery of $X_{k}$ for $0\leqslant k \leqslant n$ given the observations $Y_{0:n}$, a standard estimator is $\pE^{\parvec}[X_{k}\vert Y_{0:n}]$ which corresponds to $\addf{k}(x_{k},x_{k+1}) = x_k$. In Expectation Maximization-based MLE estimation, the intermediate quantity 
$\parvec\mapsto Q(\parvec,\parvec') = \pE^{\parvec'}[\sum_{k=0}^{n-1} \log \qg{k}{\parvec}(X_{k}, X_{k+1}) | Y_{0:n}]$
is another example where $\addf{k}(x_{k},x_{k+1}) = \log \qg{k}{\parvec}(x_{k}, x_{k+1})$. Recursive MLE (RMLE) methods express $\deriv \log \llh{n}{\parvec} = \pE^\theta[\sum_{k=0}^{n-1} \deriv \log \qg{k}{\parvec}(X_{k},X_{k+1}) | Y_{0:n}]$ via Fisher's identity under some regularity conditions (see \cite{cappe2005inference}, Chap. 10), in which case $\addf{k}(x_{k},x_{k+1}) = \deriv \log \qg{k}{\parvec}(x_{k},x_{k+1})$.

The challenge of computing \eqref{eq:smoothing:expectation} is twofold, i) the smoothing distribution is generally intractable, ii) under this distribution, expectations are also intractable. A classical approach is to learn both the distribution and expectations using Markov chain or sequential Monte Carlo methods, (see \cite{chopin2020introduction}, Chapter 12, for a recent review of SMC methods). 
In the case of additive functionals, more recent generic estimators based on SMC have been designed \cite{mastrototaro2021fast, martin2022backward}, and their theoretical properties (consistency, asymptotic variance and normality) have been studied \cite{gloaguen2022pseudo}. 
However, Monte Carlo methods show limitations when the dimension $d$ of the latent space is large, and alternatives using variational inference are appealing and computationally efficient solutions.

\subsection{Variational inference for sequential data}

In variational approaches, instead of designing Monte Carlo estimators of $\post{0:n}^{\parvec} h$ (or of the conditional distribution of the states given the observations), the conditional law $\post{0:n}^{\parvec}$ of $X_{0:n}$ given $Y_{0:n}$ is approximated by choosing a candidate in a parametric family $\{ \vd{0:n}\}_{\parvar \in \Lambda}$, referred to as the \textit{variational} family, where $\Lambda$ is a parameter set. 
Parameters are then learned by maximizing the \textit{evidence lower bound} (ELBO) defined as:
\begin{equation}
\label{eq:ELBO}
\mathcal{L}(\parvec,\parvar) = \pE_{ \vd{0:n}}\left[\log \frac{p^{\parvec}_{0:n}(X_{0:n},Y_{0:n})}{ \vd{0:n}(X_{0:n})}\right] = \int \log \frac{p^{\parvec}_{0:n}(x_{0:n},Y_{0:n})}{ \vd{0:n}(x_{0:n})} \vd{0:n}(x_{0:n})\mu(\rmd x_{0:n})\eqsp,
\end{equation}
where $p^{\parvec}_{0:n}$ is the joint probability density function of $(X_{0:n},Y_{0:n})$ when the model is parametrized by $\parvec$. A critical point therefore lies in the form of the variational family. Motivated by the sequential nature of the data, most works impose further structure on the variational family via a factorized decomposition of $\vd{0:n}$ over $x_{0:n}$ \cite{johnson2016, krishnan2017structured, Lin2018VariationalMP, Marino2018AGM}. Here, the natural strategy is to reintroduce part or all of the conditional independence properties of the true generative model.

\subsection{Backward factorization of the smoothing distribution}
\label{sec:backward_fact}
Under the true model, the \textit{filtering} distribution at time $k$ is defined as the distribution of $X_k$ given $Y_{0:k}$, with density w.r.t the Lebesgue measure denoted as $\post{k}^{\parvec}$. 
One known factorization of $\post{0:n}^{\parvec}$ - albeit not used in the aforementioned works - exists by further introducing the distribution of the so-called \textit{backward kernels}, that is, for each $0\leqslant k \leqslant n - 1$, the conditional distribution of $X_k$ given $(X_{k+1}, Y_{0:k})$ whose density is proportional to $x_k \mapsto \hd{k}{\parvec}(x_{k},x_{k+1})\post{k}^{\parvec}(x_{k})$. 
A key result for SSMs is that, conditionally on the observations, the reverse-time process $(X_{n-k})_{0\leqslant k\leqslant n}$ is an \textit{inhomogeneous} Markov chain whose initial distribution is the filtering distribution at $n$, and whose transition kernels are precisely the backward kernels. 
This allows the following \textit{backward factorization}:
\begin{equation*}
    \post{0:n}^{\parvec}(x_{0:n}) = \post{n}^{\parvec}(x_n)\prod_{k=1}^{n}\frac{\hd{k - 1}{\parvec}(x_{k - 1},x_{k})\post{k - 1}^{\parvec}(x_{k})}{\int \hd{k - 1}{\parvec}(x,x_{k})\post{k - 1}^{\parvec}(x) \mu(\rmd x)}\eqsp. 
\end{equation*}

Since each backward kernel at time $k$ only depends on observations up to time $k$, a major practical advantage of this decomposition is to allow \textit{recursive} estimation of the smoothing distributions: when a new observation $Y_{k+1}$ is processed, obtaining $\post{0:k+1}^{\parvec}$ only amounts to computing $\post{k+1}^{\parvec}$ and the associated backward kernel, while previous terms in the product stay fixed. 
Recently, \cite{campbell2021online} proposed  a related variational family by introducing
\begin{equation}
    \label{eq:varpost:backward:factorization}
 \vd{0:n}(x_{0:n})=  \vd{n}(x_n)\prod_{k=1}^{n}\vd{k-1\vert k}(x_{k},x_{k - 1})\eqsp,
\end{equation}
where $\vd{n}$ (resp. $\vd{k - 1\vert k}(x_{k},\cdot)$) are user-chosen p.d.f. whose parameters typically would depend on $Y_{0:n}$ (resp. $Y_{0:k}$). Under \eqref{eq:varpost:backward:factorization}, the ELBO \eqref{eq:ELBO} becomes an expectation of an additive functional.

\section{A control on backward variational additive smoothing}
\label{sec:theoretical_result}
In the context where the variational factorization follows \ref{eq:varpost:backward:factorization}, we now present our main theoretical result. 

For all $x_k\in\mathbb{R}^d$ and $\theta\in\Theta$, define $\uk{k}{\parvec}(x_k, \cdot)$ the kernel with density $\qg{k}{\parvec}(x_{k},\cdot) $ with respect to the Lebesgue measure:
$$
\uk{k}{\parvec}(x_k, \rmd x_{k+ 1})  = \hd{k}{\parvec}(x_{k}, x_{k+1})\md{k+1}{\parvec}(x_{k+1},Y_{k+1})\mu(\rmd x_{k+1})\eqsp.
$$
\begin{hypH}
\label{assum:bias:bound}
There exist distributions $\ivd{k}$, $\parvar\in\parvarspace$, and  functions $c_k$, $0\leqslant k \leqslant n$, such that $\ivd{n}=\vd{n}$ and for all $1\leqslant k \leqslant n$, $\parvec\in\parspace$, $\parvar\in\parvarspace$, all bounded measurable functions $h$ on $\mathbb{R}^d\times\mathbb{R}^d$, 
$$
\left|\ivd{k}\vd{k-1 \vert k} h - \frac{\ivd{k-1} \uk{k-1}{\parvec}h}{\ivd{k-1} \uk{k-1}{\parvec}\1}\right| \leqslant c_k(\parvec,\parvar) \| h \|_\infty,
$$ 
and for all bounded measurable functions $h$ on $\mathbb{R}^d$,
$$
\left|\ivd{0} h - \post{0}^{\parvec}h \right|\leqslant c_0(\parvec,\parvar)\|h\|_\infty\eqsp,
$$
where  $\post{0}^{\parvec}$ is the filtering distribution at time $0$, i.e. $\post{0}^{\parvec}h = \chi^\parvec \md{0}{\parvec}h/\chi^\parvec \md{0}{\parvec}\1$.
\end{hypH}
Note that under H\ref{assum:bias:bound}, choosing $h$ such that there exists $\tilde h$ satisfying  $h:(x_{k-1},x_k) \mapsto \tilde h (x_k)$, yields for all $\parvec\in\parspace$, $\parvar\in\parvarspace$, 
$$
\left|\ivd{k} \tilde h - \frac{\ivd{k-1} \uk{k-1}{\parvec} \tilde h}{\ivd{k-1} \uk{k-1}{\parvec}\1}\right| \leqslant c_k(\parvec,\parvar) \| \tilde h \|_\infty.
$$ 

\begin{hypH}
\label{assum:strong:mixing}
There exist constants $0 < \udlow < \udup < \infty$ such that for all $k \in \nset$, $\parvec\in\Theta$, $\parvar\in\parvarspace$ and $(x_k, x_{k + 1}) \in \mathbb{R}^d \times \mathbb{R}^d$, 
$$
    \udlow \leq \ud{k}^{\parvec}(x_k, x_{k + 1}) \leq \udup
$$ 
and 
$$
    \udlow \leq  \vd{k\vert k+1}(x_{k+1}, x_{k}) \leq \udup. 
$$ 
\end{hypH}

\begin{proposition}
\label{prop:bias}
Assume that H\ref{assum:bias:bound} and H\ref{assum:strong:mixing} hold. Then, for all $n \in \nset$, $\parvec\in\Theta$, $\parvar\in\parvarspace$, and all additive functionals $h_{0:n}$ as in \eqref{eq:additive:functional},  
\begin{multline*}
        \big| \vd{0:n} h_{0:n} -  \post{0:n}^{\parvec} h_{0:n} \big| 
        \leq 2\frac{\sigma_+}{\sigma_-}\sum_{k=0}^{n-1}\left\| \addf{k}\right\|_\infty \\
        \times\left( c_0(\parvec,\parvar) + \sum_{m=1}^k \rho^{k-m+1}c_m(\parvec,\parvar) + c_{k+1}(\parvec,\parvar) + \sum_{m=k+2}^n \rho^{m-k-1}c_m(\parvec,\parvar)\right)\eqsp,
\end{multline*}
where $\rho = 1-\sigma_-/\sigma_+$ and where $\sigma_-$ and $\sigma_+$ are defined in H\ref{assum:strong:mixing}.
\end{proposition}

\begin{proof}
The proof is postponed to Appendix~\ref{sec:proof}.
\end{proof}
By Proposition~\ref{prop:bias}, if there exist $h_\infty$ and $c_+$ such that for all $0\leq k \leq n-1$, $\|\addf{k}\|_\infty\leq h_\infty$ and for all $\parvec\in\parspace$, $\parvar\in\parvarspace$, $0\leq m\leq n$, $c_m(\parvec,\parvar)\leq c_+(\parvec,\parvar)$ then
\begin{equation}
\label{eq:additive:error}
\big| \vd{0:n} h_{0:n} -  \post{0:n}^{\parvec} h_{0:n} \big| 
        \leq 4\frac{\sigma_+}{\sigma_-}\left(1 + \frac{\rho}{1-\rho}\right)c_+(\parvec,\parvar)h_\infty n\eqsp.
\end{equation}
On the other hand, if we are interested in marginal smoothing distributions, i.e. cases where $\addf{j} = 0$ for all $j\neq k$, Proposition~\ref{prop:bias} yields a uniform control in time:
$$
\big| \vd{0:n} \addf{k} -  \post{0:n}^{\parvec} \addf{k} \big| 
        \leq 4\frac{\sigma_+}{\sigma_-}\left(1 + \frac{\rho}{1-\rho}\right)c_+(\parvec,\parvar)h_\infty \eqsp.
$$

\subsection{Comments on assumptions H\ref{assum:bias:bound} and H~\ref{assum:strong:mixing}}
\label{sec:assumptions}
Assumption \textbf{H\ref{assum:bias:bound}} is a pivotal technical tool to prove Proposition \ref{prop:bias}. 
Nonetheless, it is not a strong assumption as for any sequence of distributions $\left(\ivd{k}\right)_{1\leqslant k \leqslant n}$, the sequence $c_k(\parvec,\parvar)$ can be chosen to be the total variation between  $(x_{k-1}, x_k) \mapsto \ivd{k}(x_k)\vd{k-1 \vert k}(x_{k}, x_{k-1})$ and the probability density proportional to
$(x_{k-1}, x_k) \mapsto \ivd{k-1}(x_{k-1}) \ud{k-1}^{\parvec}(x_{k-1}, x_k)$. 
However, a challenging task for future research would be to find the best sequence of $\ivd{k}$ in terms of $c_k(\parvec,\parvar)$.
We now show that in some specific examples, given a sequence of $\ivd{k}$, an explicit sequence of $c_k(\parvec, \parvar)$ can be given.

\paragraph{Exact inference.} It is worth noting that if  $\vd{n}$ is the true filtering distribution at time $n$ and $(\vd{k-1\vert k})_{k\geqslant 1}$ are  the true backward distributions, then the unique sequence $(\ivd{k})_{k\geqslant 1}$ that achieves $c_k(\theta, \lambda) = 0$ in \textbf{H\ref{assum:bias:bound}} for all $k$ is the sequence of true filtering distributions.

\paragraph{Linear and Gaussian case.}
In  the linear and Gaussian case, we assume that for all $x_{k-1}$, $\hd{k-1}{\parvec}(x_{k-1}, \cdot)$ is the Gaussian p.d.f with mean $A_{k}^{\parvec} x_{k-1}$ and variance $R_{k}^{\parvec}$  and that $\md{k}{\parvec}(x_{k},\cdot)$ is the Gaussian p.d.f with mean $B_{k}^{\parvec} x_k$ and variance $S_{k}^{\parvec}$. 

In this setting, assume that $\ivd{k}$ is the Gaussian p.d.f. with mean $\mu_{k}^{\parvar}$ and variance $\Sigma_{k}^{\parvar}$ and that for all $x_k$, $\vd{k-1\vert k}(x_{k}, \cdot)$ is the Gaussian p.d.f with mean $A_{k}^{\parvar} x_k$ and variance $R_{k}^{\parvar}$. Therefore, we assume (i) that the variational backward kernels are linear as are the backward kernels of the true model and (ii) that the instrumental intermediate distributions $\ivd{k}$, $0\leq k \leq n-1$, are Gaussian as the filtering distributions of the true model. As described below, this choice allows to obtain an explicit upper bound for $c_k$, $0\leq k \leq n$. This highlights that assumption \textbf{H\ref{assum:bias:bound}} can be made usable in practice. This also emphasizes the versatility of \textbf{H\ref{assum:bias:bound}} as other instrumental densities could be tuned, since this specific choice is not proved to be optimal.

Choosing $\jvd_{k-1:k}^\parvar$  (resp. $\jvd_{k-1:k}^{\parvar,\parvec}$) as a  short-hand notation for the joint distribution $\ivd{k}\vd{k-1 \vert k} $ (resp. $\ivd{k-1} \uk{k-1}{\parvec}h/\ivd{k-1} \uk{k-1}{\parvec}\1$), standard computations show that $\jvd_{k-1:k}^\parvar$ (resp. $\jvd_{k-1:k}^{\parvar, \parvec}$) is a  multivariate Gaussian distributions with known mean $\mathsf{M}_{k}^{\parvar}$ (resp. $\mathsf{M}_{k}^{\parvar,\parvec}$) and variance 
$\mathsf{V}_{k}^{\parvar}$ (resp. $\mathsf{V}_{k}^{\parvar,\parvec})$. In this case, for all bounded and measurable function $h$,
$$
\left|\ivd{k}\vd{k-1 \vert k} h - \frac{\ivd{k-1} \uk{k-1}{\parvec}h}{\ivd{k-1} \uk{k-1}{\parvec}\1}\right| \leqslant 2 \left\|\jvd_{k-1:k}^\parvar - \jvd_{k-1:k}^{\parvar,\parvec}\right\|_{\mathrm{tv}} \| h \|_\infty,
$$ 
where $\|\cdot\|_{\mathrm{tv}}$ is the total variation distance. 
Therefore, we can choose $c_k(\parvec,\parvar) = 2 \|\jvd_{k-1:k}^\parvar - \jvd_{k-1:k}^{\parvar,\parvec}\|_{\mathrm{tv}}$. 
It remains to use the fact that $\jvd_{k-1:k}^\parvar$  and $\jvd_{k-1:k}^{\parvar,\parvec}$ are Gaussian distributions, that $\|\jvd_{k-1:k}^\parvar - \jvd_{k-1:k}^{\parvar,\parvec}\|_{\mathrm{tv}} \leqslant (\mathrm{KL}(\jvd_{k-1:k}^\parvar \| \jvd_{k-1:k}^{\parvar,\parvec})/2)^{1/2}$ and that we have an explicit expression of the $\mathrm{KL}$ divergence between Gaussian  distributions which yields
\begin{multline*}
c_k(\parvec,\parvar) \propto \left(\log \left| \mathsf{V}_{k}^{\parvar,\parvec}\right|/\left| \mathsf{V}_{k}^{\parvar}\right| + \left(\Delta_k^{\parvar,\parvec}\right)^\top(\mathsf{V}_{k}^{\parvar,\parvec})^{-1}\left(\Delta_k^{\parvar,\parvec}\right)
 + \mathrm{Tr}\left( (\mathsf{V}_{k}^{\parvar, \parvec})^{-1}\mathsf{V}_{k}^{\parvar}\right) - d\right)^{-1/2}\eqsp,
\end{multline*}
where $\Delta_k^{\parvar,\parvec} = \mathsf{M}_{k}^{\parvar} - \mathsf{M}_{k}^{\parvar,\parvec}$ , $ \mathrm{Tr}$ is the Trace operator and $\propto$ means up to a multiplicative constant independent of $\parvec$ and $\parvar$.

\paragraph{About \textbf{H\ref{assum:strong:mixing}}}
This assumption is rather strong, but typically satisfied in models where the state space  is compact. This assumption is classic in the SMC literature in order to obtain quantitative bounds for errors or variance of estimators.

\section{Recursive backward variational learning with amortizing networks} 
\label{sec:algorithm}
Written as in $(\ref{eq:varpost:backward:factorization})$, the backward factorization of the variational family only imposes dependencies between the latent states. 
This minimal setup, sufficient to derive the theoretical results above, leaves a lot of freedom for implementation. 

\subsection{Amortized parametrization of the variational distribution}

First, suppose that we want to learn the variational parameters by computing ELBO gradients on sequences of fixed length $n$. Implementing \eqref{eq:varpost:backward:factorization} requires to define $n + 1$ distributions $(\vd{k-1\vert k})_{0\leq k \leq n}$ and $\vd{n}$. 
A direct approach would be to freely parameterize these distributions.
In this case, the number of parameters to learn would grow linearly with $n$, which is prohibitive for long sequences. 
To reduce the computational burden, a popular alternative is amortized inference, which in this context amounts to output the parameters of each kernel via a common highly expressive mapping, -- typically, a DNN which itself holds a fixed number of parameters. 

For this purpose, an appealing property of the backward kernels of the true data model is the incremental dependency on the observations (see  Section \ref{sec:backward_fact}). 
Indeed, the backward distribution of $x_{k-1}$ depends on  the observations up to time $Y_{0:k-1}$ (through the filtering distribution at time $k-1$) and on the state $x_k$. 
Our first step is then to encode sequentially the dependencies on the  observations through a recurrent neural network $f_{k-1}^{\parvar}(y_{0:k-1})$, such that parameters of $\vd{k - 1\vert k}$ are given by a non linear function $g_k^\parvar\left(f_{k-1}^{\parvar}(y_{0:k-1}), x_k\right)$. 
Finally, parameters of $\vd{n}$ are given by a last non linear function $f_n^\parvar\left( y_{0:n}\right)$, that typically would depend on  $f_{n-1}^{\parvar}(y_{0:n-1})$.   

\subsection{Variational recursions and online computation of the ELBO}
In the setting presented above an interesting implementation choice is when the RNN $f_{k}^\parvar$ outputs the parameters of a p.d.f. 
The RNN therefore indirectly outputs a sequence of distributions $\left(\vd{k} \right)_{1\leqslant k \leqslant n}$. 
These distributions can be used at each time $k$ to define online variational distributions that factorize as in \eqref{eq:varpost:backward:factorization} (replacing $\vd{n}$ by $\vd{k}$). From there, the ELBO can be computed online. 
Indeed, note that at time $n$, using the tower property of expectations,  $\mathcal{L}(\parvec,\parvar) = \pE_{ \vd{n}}[T_n(X_n)]$ where $T_n(X_n) = \pE_{\vd{0:n}}[\log p^{\parvec}_{0:n}(X_{0:n},Y_{0:n}) /\vd{0:n}(X_{0:n}) \vert X_n ]$. 
This statistic can be computed recursively, since, for all $k \geq 0$,
\begin{equation}
    \label{ref:eq:v_t}
    T_k(X_k) = \pE_{\vd{k-1 \vert k}}\left[T_{k-1}(X_{k-1}) + \log \frac{\qg{k}{\parvec}(X_{k-1}, X_k) \vd{k-1}(X_{k-1})} {\vd{k-1 \vert k}(X_k, X_{k-1})  \vd{k}(X_k)} \middle\vert X_k \right]\,.
\end{equation}
A more detailed derivation is provided in the appendix. 
An important point is that contrary to \cite{campbell2021online}, we only assume that each of the \textit{joint} distributions $(\vd{0:k})_{k \geq 0}$ is an approximation of $(\post{0:k}^\parvec)_{k \geq 0}$. Interestingly, \textbf{H\ref{assum:bias:bound}} hints that best results may be obtained by actually enforcing that $(\vd{k})_{k \geq 0}$ and $(\vd{k-1 \vert k})_{k \geq 0}$ approximate the densities of the true filtering and backward distributions. Still, we find that competitve results are obtained without further regularization, even in the amortized setting where the global set of parameters is optimized jointly over time.

\section{Numerical experiments}
\label{sec:experiments}
\subsection{Linear Gaussian SSMs and equality in H\ref{assum:bias:bound}}
\label{sec:experiments:linear_gaussian}
First, we want to study empirically the special case where the variational family \textit{contains} the true model. This can be achieved when the true state-space model  is a linear and Gaussian SSM, i.e. when $\chi^\parvec$ (resp. $\hd{k}{\parvec}(X_k, \cdot)$ and $\md{k}{\parvec}(X_k,\cdot)$) are densities of Gaussian distributions with mean $A_0$ (resp. $A X_k$ and $B X_k$) and variance $Q_0$ (resp. $Q$ and $R$), such that $\theta = (A_0, Q_0, A, Q, B, R)$. If we define $(\vd{k-1 \vert k})_{k \leq n}$ and $\vd{n}$ as the backward and filtering densities of a similar model with parameters $\lambda = (\bar{A_0}, \bar{Q_0}, \bar{A}, \bar{Q}, \bar{B}, \bar{R})$, then $\vd{0:n} = \post{0:n}^\parvec$ for $\lambda = \theta$. When this is achieved, Section \ref{sec:assumptions} shows that $c_k(\parvec, \lambda) = 0$ for all $k$, suggesting that the additive error vanishes. In this setting, the form and parameters of the variational backward and filtering kernel is given analytically via the Kalman filtering and smoothing recursions, thus the computations of $\post{0:n}^\parvec$, $\vd{0:n}$ and all expectations in (\ref{ref:eq:v_t}) are fully tractable. In this  example, the parameter $\parvec$ is known and $\parvar$ is trained in the case $d=1$ and using samples of $n=64$ observations. The training curve is given in Figure~\ref{fig:linear_gaussian_training_curve}.

In Figure \ref{fig:linear_gaussian_additive_error}, we depict the controlled term of Proposition \ref{prop:bias} in the case of state estimation, i.e. for $h_{0:n}: x_{0:n} \mapsto \sum_{k=0}^n x_k$. This is done by sampling $J=20$ observation sequences $(Y^j_{0:n})_{1\leq j \leq J}$ of length $n=2000$ using the true model with parameter $\parvec$. This clearly illustrates the linear dependency on the number of observations. We also find that the error rates can vary greatly between parameters $\lambda_1 \neq \lambda_2$, even when $|\mathcal{L}(\theta,\lambda_1) - \mathcal{L}(\theta,\lambda_2)|$ is small. This is observed by computing the errors for different stopping points of the optimization. Sampling distinct sequences $(Y^j_{0:n})_{1\leq j \leq J}$  highlights the dependency of $(c_k(\theta,\lambda))_{0\leq k \leq n}$ on the observations. In the appendix, we provide more implementation details, as well as additional figures for the errors on the marginal distributions.

\begin{figure}
    \begin{subfigure}{0.5\textwidth}
      \centering
      \includegraphics[width=\linewidth]{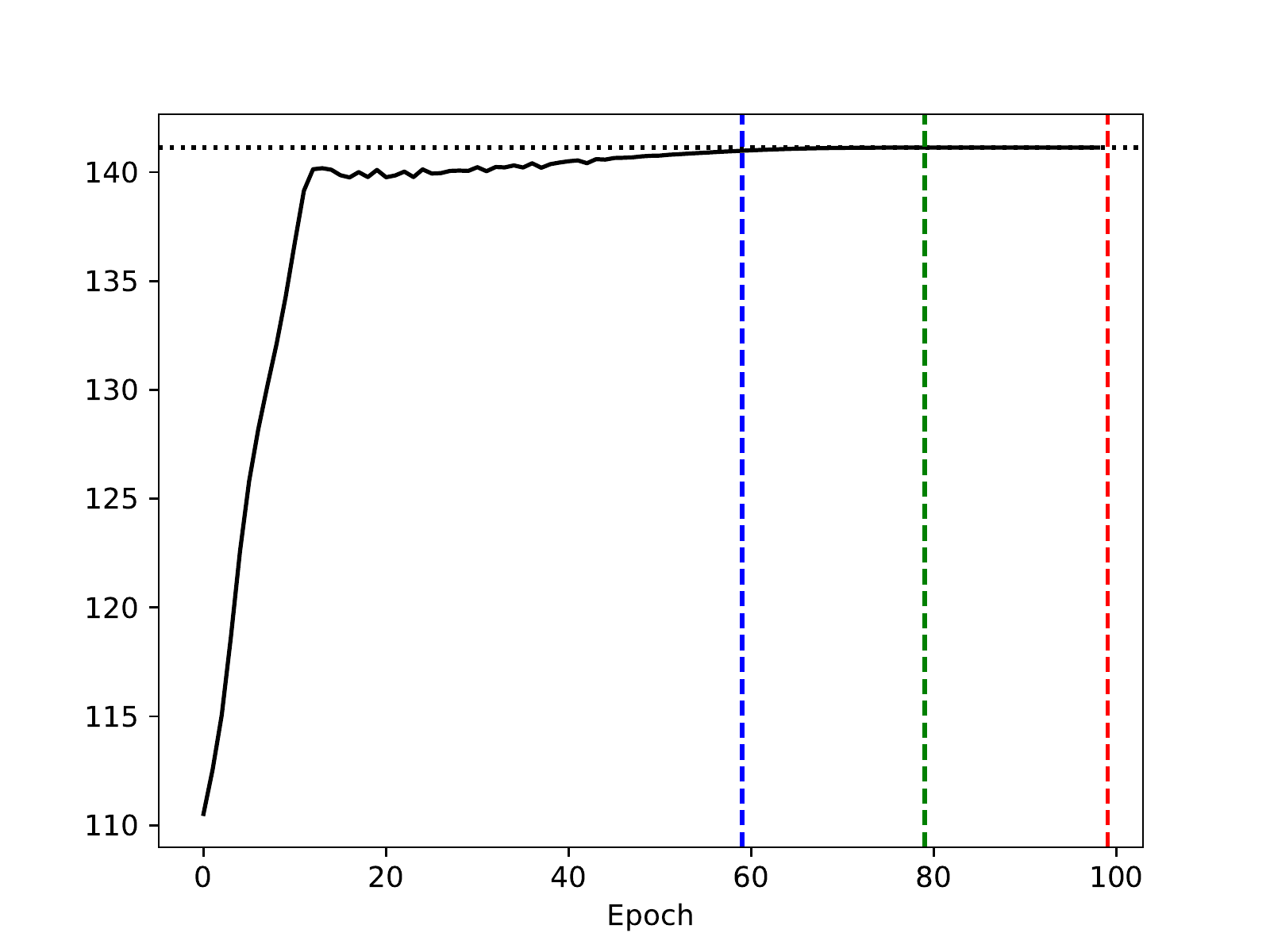}
      \caption{$\llh{n}{\parvec}$ (dotted line) and $\lambda \mapsto\mathcal{L}(\theta, \lambda)$ over epochs (full line).}
      \label{fig:linear_gaussian_training_curve}
    \end{subfigure}
    \begin{subfigure}{0.5\textwidth}
      \centering
      \includegraphics[width=\linewidth, height=0.8\textwidth]{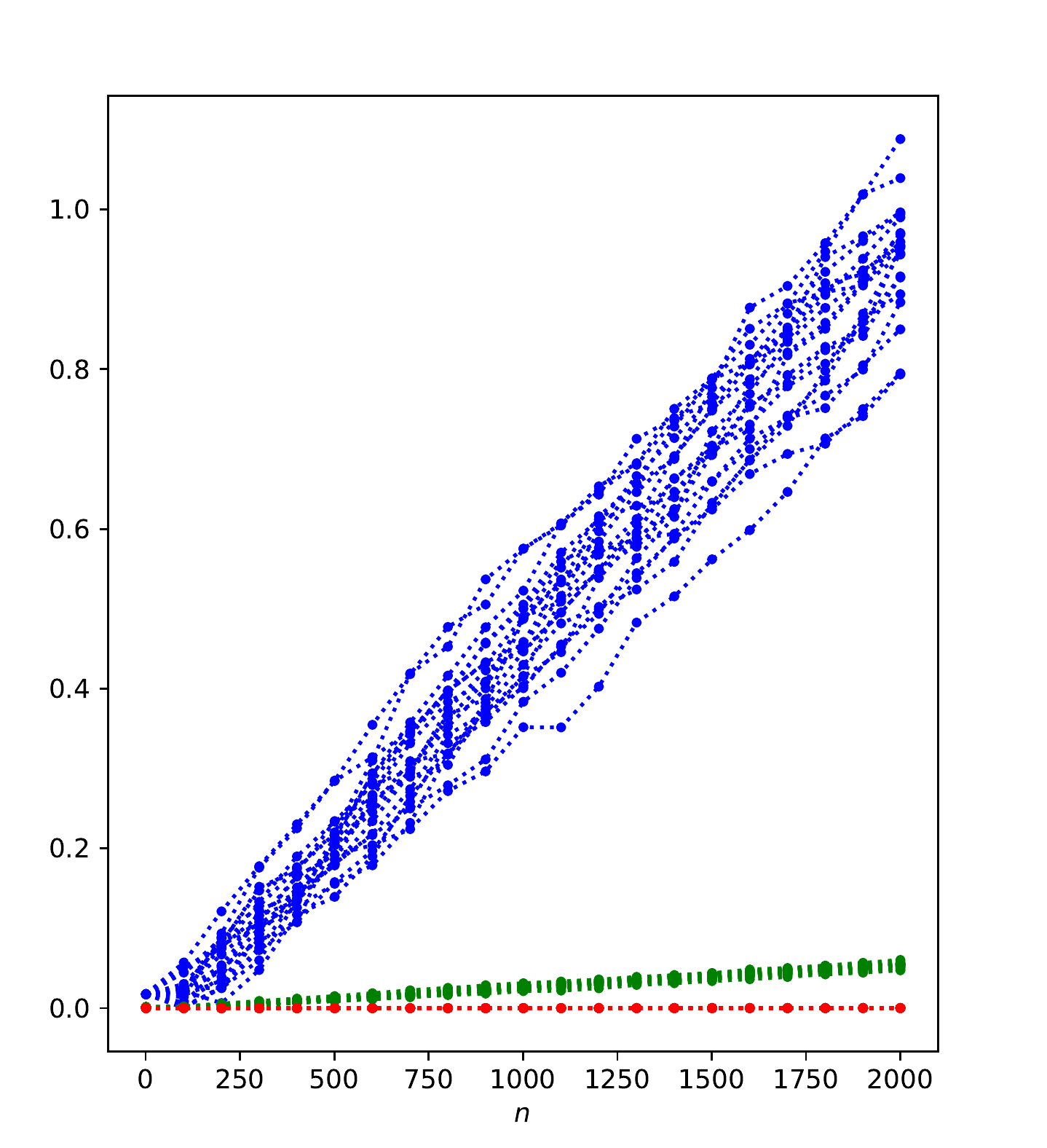}
      \caption{Smoothing errors between the variational model and the true model, i.e. $\big| \vd{0:n} h_{0:n} -  \post{0:n}^{\parvec} h_{0:n} \big|$ for $\tilde{h}_{k}(x_k, x_{k+1}) = x_k$;}
      \label{fig:linear_gaussian_additive_error}
    \end{subfigure}

    \caption{ ELBO during the training of $\parvar$ (left).  Additive smoothing error for a linear Gaussian variational model at successive stopping points of the optimization (blue, green and red), on $J=20$ different observation sequences (right).}
    \label{fig:linear_gaussian}
\end{figure}

\subsection{Expressive capabilities of backward variational families in nonlinear Gaussian SSMs}
\label{sec:experiments:nonlinear}

We now consider a generative model where the prior distribution and transition kernels are still linear, but $\md{k}{\parvec}(X_k,\cdot)$ is the Gaussian probability density with mean $h^\theta(X_k)$ and variance $R$, $h^\theta$ being a nonlinear mapping commonly referred to as the \textit{decoder}. In this setting, \cite{Hlv2021DisentanglingIF} showed for the first time that no assumptions are required on $h^\theta$ for identifiable state estimation. The authors obtained promising results via a variational approximation $\vd{0:n}$ which can be analytically marginalized and therefore allows fast inference. We briefly explain how this variational approximation can be generalized in our context. For all $k \geq 0$, $\vd{k}$ (resp. $\vd{k-1|k}(X_k,\cdot)$) is a Gaussian probability density with mean $\mu_k$ (resp. $\backward{A}_k X_k + \backward{a}_k$) and variance $\Sigma_k$ (resp.  $\backward{\Sigma}_k$).  Moreover, a variational prior $\bar{\chi}^\parvar$ and variational transition kernels $\bar{m}_k^\parvar(X_k, \cdot)$ are introduced as Gaussian densities with mean $\bar{A}_0$ (resp. $\bar{A} X_k$) and variance $\bar{Q}_0$ (resp. $\bar{Q}$) which enforces hidden dynamics of the variational model to have the same form as the data model. 
We then suppose that:
\begin{itemize}
    \item $\left(\mu_k, \Sigma_k\right) = r^\lambda(u_k, y_k)$, where $u_k = (\bar{A}\mu_{k-1}, \bar{A} \Sigma_{k-1} \bar{A}^T + \bar{Q})$ and $r^\lambda$ is a mapping to be specified below.
    \item $\vd{k-1|k}(X_k, X_{k-1}) \propto  \bar{m}_k^\lambda(X_{k-1}, X_k) \vd{k}(X_k)$.
\end{itemize}
The linear dynamics of $\bar{m}_k^\lambda(X_k, \cdot)$ prescribe Kalman-type \textit{predict} and \textit{backward} updates, while $u_k$ are the parameters of an intermediate \textit{predictive} Gaussian distribution. The mapping $r^\lambda$ then performs the Bayesian \textit{update} step and can be of any form. 
In \cite{Hlv2021DisentanglingIF}, the authors do not use of a generic form for this update step but follow \cite{johnson2016} and impose that $\mu_k,\Sigma_k$ is the result of the conjugation of two Gaussian distributions: the predictive whose parameters are $u_k$, and a variational approximation of $x_k \mapsto \md{k}{\parvec}(x_k, y_k)$ whose parameters are given by a DNN $f_{\mathrm{enc}}^\lambda$ (referred to as the \textit{encoder}) which takes only $y_k$ as input. While this form is required for tractable inference in their framework (as they build $\vd{0:n}$ from it with the sum-product algorithm for SSMs)  our backward formulation does not require this, and we show that higher performance can be obtained by letting a DNN $r^\lambda$ learn a more realistic conjugation of new observations with the running variational filtering estimates. 

In this context, the true smoothing distribution $\post{0:n}^\parvec$ has no analytic form. As a surrogate for this ground truth, we use the particle-based Forward Filtering Backward Simulation (FFBSi) algorithm. The FFBSi outputs  trajectories (here, $1000$ samples) approximately sampled from the true target smoothing distributions using sequential importance sampling and resampling steps. This algorithm is also based on a forward-backward decomposition of the smoothing distributions (see \cite{douc2014nonlinear}, Chapter 11, for details). 
We remain in the case $d = 1$ to ensure that this approximation is good. We provide additional implementation details and figures in the appendix.

In the case where $h^\theta$ is a non-injective mapping, we compare the additive error with respect to the FFBSi (i.e. the left hand term of equation \eqref{eq:additive:error}) obtained for our parametrization and the one of \cite{Hlv2021DisentanglingIF} for $h_{0:n}: x_{0:n} \mapsto \sum_{k=0}^n x_k$. 
Figure \ref{fig:additive_error_nonlinear} shows that our method reduces significantly this error. In Figure \ref{fig:table_errors}, we report the quality of the FFBSi estimator in the form of the sample mean and variance of its error against the true states. We then report the final additive smoothing errors of the variational methods after processing all of the $n=500$ observations of the evaluation sequences. The results confirm our intuition that our framework leads to more expressive variational distributions, especially when the distribution of $X_k$ given $Y_k$ admits several modes. Indeed, the framework of \cite{Hlv2021DisentanglingIF} approximates $x_k \mapsto p^\theta(x_k|y_k)$ by an encoder   $f_{enc}^\lambda(y_k)$ that outputs a Gaussian density. 
In contrast, our parametrization only assumes Gaussianity for the variational filtering distribution and does not attempt to solve the inverse problem of modeling the distribution of $X_k$ given $Y_k$ without the dynamics. 
Therefore, here, the backwards formulation allows to conserve analytical marginalisation of $\vd{0:n}$ without modeling the previous distribution as an intermediate step, which increases performance.

\begin{figure}
      \includegraphics[width=\linewidth, trim={0 0 0 1cm}, clip]{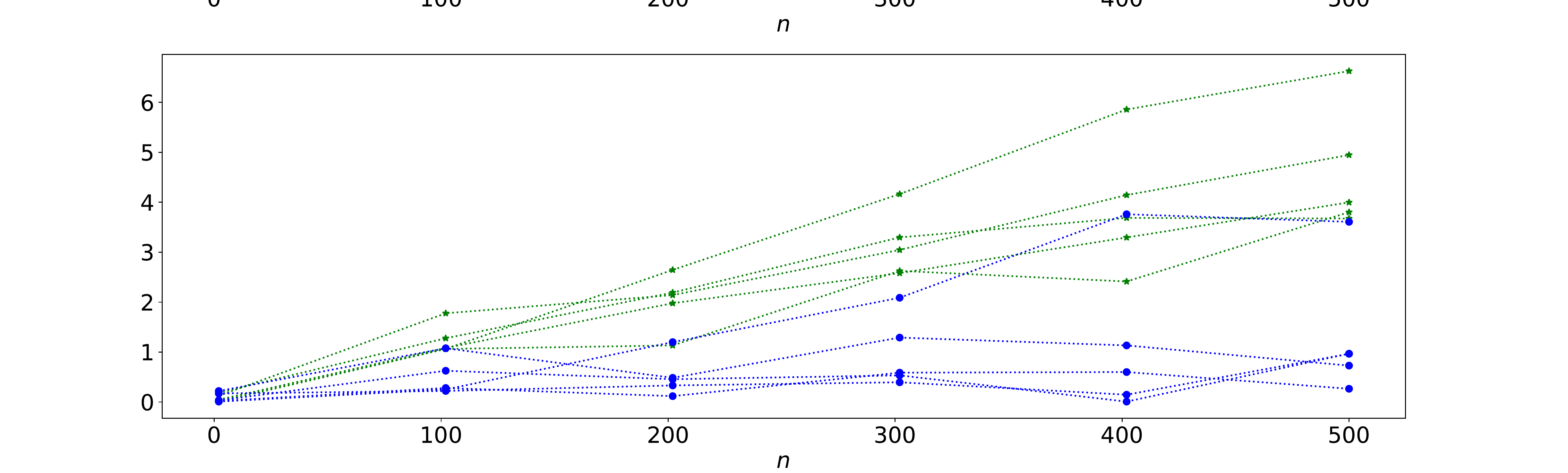}
      \caption{Smoothing errors $\big| \vd{0:n} h_{0:n} -  \post{0:n}^{\parvec} h_{0:n} \big|$ for $\tilde{h}_{k}(x_k, x_{k+1}) = x_k$, with our variational approach (blue dots) and that of \cite{Hlv2021DisentanglingIF} (green stars). Experiments were produced on 5 independent (simulated) data set, hence the 5 replicates.}
      \label{fig:additive_error_nonlinear}
\end{figure}

\begin{figure}
    \centering
    \begin{tabular}{ccccc}
        \toprule
        {Seq.} & $\mathrm{{Mean~err._{FFBSi}}}$ &   $\mathrm{{Var~err._{FFBSi}}}$ &  $\mathrm{Smooth~err._{FFBSi / Johnson}}$  &    $\mathrm{Smooth~err._{FFBSi / Ours}}$ \\
        \midrule
        0 & 0.05 & 0.01 & 4.95 & 0.73 \\
        1 & 0.04 & 0.00 & 3.80 & 0.97 \\
        2 & 0.05 & 0.01 & 4.00 & 0.27 \\
        3 & 0.03 & 0.00 & 3.67 & 0.97 \\
        4 & 0.07 & 0.02 & 6.63 & 3.61 \\
        \bottomrule
    \end{tabular}
    \caption{First column: empirical mean of $\{(\hat{x}_{k, FFBSi} - x_k^{*})^2\}_{0 \leqslant k \leqslant n}$ where $x_k^{*}$ is the true state and $\hat{x}_{k, FFBSi}$ is the marginal mean of $\post{0:n}^\parvec$ at time $k$ provided by the FFBSi algorithm. Second column: empirical variance of the same quantity. Third and fourth column: smoothing errors $\big| \vd{0:n} h_{0:n} -  \post{0:n}^{\parvec} h_{0:n} \big|$ for $\tilde{h}_{k}(x_k, x_{k+1}) = x_k$ of the two compared methods at time  $n=500$ when $\post{0:n}^\parvec$ is given by the FFBSi algorithm.}
    \label{fig:table_errors}
    \label{table:errors}

\end{figure}

\section{Discussion}
We have provided the first bound on the additive smoothing error in the context of sequential variational inference using a backward factorization. We have empirically presented clear cases to highlight the practical consequences of this theoretical result. We have also shown that existing methods can be reframed into filtering and backward recursions: in this case, we found that more flexible updates are available without increasing the computational workload. Some limitations of our work and challenges for further research are the following.
\begin{itemize}
    \item Our theoretical result sheds light on important properties of sequential variational methods, but the assumptions involved are not fully \textit{constructive}, i.e. we believe that further works may provide more explicitly the form of the optimal variational factors under given parametric families of the variational kernels. 
    \item Empirically, we have restricted to the case where analytical computations are available to marginalize the joint variational smoothing distribution. More computationally heavy approaches requiring Monte Carlo sampling for marginalisation are possible, and may further improve the state estimation results shown in Section \ref{sec:experiments}.
    \item Since the DNNs involved in our implementation take the estimations of the current dynamics as input, we find that training in our context suffers more easily from the drawbacks of gradient descent in recurrent models, e.g. it is more amenable to vanishing / exploding gradients.
\end{itemize}

As a novel variational approach for sequential data, this work has potential applications in many areas. This work does not present any foreseeable societal consequence.

\bibliographystyle{apalike}
\bibliography{backward_variational}

\appendix
\section{Proof of Proposition~\ref{prop:bias}}
\label{sec:proof}
Following \cite{gloaguen2022pseudo}, write 
$$
\vd{0:n} h_n - \post{0:n}^\parvec h_n = \sum_{k = 0}^{n - 1} \left( \vd{0:n} \addf[n]{k} - \post{0:n}^\parvec \addf[n]{k} \right), 
$$
where, for each $k \in \intvect{0}{n - 1}$, $\addf[n]{k}$ is defined on $(\mathbb{R}^d)^{n+1}$ by
\begin{equation} \label{eq:def:addf}
\addf[n]{k} :  x_{0:n} \mapsto \addf{k}(x_k, x_{k + 1})\,. 
\end{equation}
Define, for each $n \in \nset$ and $m \in \intvect{0}{n}$, the kernel 
\begin{equation} \label{eq:def:uk:products}
    \uk{m, n}{\parvec}(x_{0:m}', \rmd x_{0:n}) \eqdef \delta_{x_{0:m}'}(\rmd x_{0:m}) \prod_{\ell = m}^{n - 1} \uk{\ell}{\parvec}(x_\ell, \rmd x_{\ell + 1}) 
\end{equation}
on $(\mathbb{R}^d)^{n + 1} \times \mathcal{B}((\mathbb{R}^d)^{n + 1})$, with the convention $\prod_{\ell = n}^{n - 1} f(\ell) = 1$ .  This yields the following decomposition:
\begin{multline*}
\vd{0:n} \addf[n]{k} - \post{0:n}^\parvec \addf[n]{k} = 
\sum_{m = 1}^n 
\left( 
\frac{\ivd{0:m} \uk{m, n}{\parvec} \addf[n]{k}}{\ivd{0:m} \uk{m, n}{\parvec} \1_{}}
- \frac{\ivd{0:m - 1} \uk{m - 1, n}{\parvec} \addf[n]{k}}{\ivd{0:m - 1} \uk{m - 1, n}{\parvec} \1_{}} 
\right) \\ 
+ \frac{\ivd{0} \uk{0, n}{\parvec} \addf[n]{k}}{\ivd{0} \uk{0, n}{\parvec} \1_{}}
- \frac{\chi^\parvec \md{0}{\parvec}\uk{0, n}{\parvec} \addf[n]{k}}{\chi^\parvec\md{0}{\parvec}\uk{0, n}{\parvec} \1_{}},
\end{multline*}
where $\ivd{0:m} = \ivd{m}\prod_{k=1}^m\vd{k - 1\vert k}$, ($1\leq m \leq n$), $\ivd{0:0} = \ivd{0} $, and since $\chi^\parvec \md{0}{\parvec} \uk{0, n}{\parvec} \addf[n]{k}/\chi^\parvec\md{0}{\parvec}\uk{0, n}{\parvec} \1_{} = \post{0:n}^\parvec \addf[n]{k}$.  For each $n \in \nset$, define  $\retrokmod_{0, n}(x_0', \rmd x_{0:n}) \eqdef \delta_{x_0'}(\rmd x_0) \,  \prod_{\ell = 0}^{n - 1} \uk{\ell}{\parvec}(x_\ell, \rmd x_{\ell + 1})$ and for $m \in \intvect{1}{n}$, 
\begin{equation} \label{eq:def:retrokmod}
    \retrokmod_{m, n}(x_m', \rmd x_{0:n}) \eqdef \delta_{x_m'}(\rmd x_m) \, \prod_{\ell = 0}^{m - 1} \kvd{k \vert k+1}(x_{\ell + 1}, \rmd x_\ell) \prod_{\ell = m}^{n - 1} \uk{\ell}{\parvec}(x_\ell, \rmd x_{\ell + 1}), 
\end{equation}
on  $\mathbb{R}^d \times \mathcal{B}((\mathbb{R}^d)^{n+1})$. 
As for all $m \in \intvect{1}{n}$ and measurable function $h$, $\ivd{0:m} \uk{m, n}{\parvec} h = \ivd{m} \retrokmod_{m, n} h$,  
$$
\frac{\ivd{0:m} \uk{m, n}{\parvec} \addf[n]{k}}{\ivd{0:m} \uk{m, n}{\parvec} \1_{}} - \frac{\ivd{0:m - 1} \uk{m - 1, n}{\parvec} \addf[n]{k}}{\ivd{0:m - 1} \uk{m - 1, n}{\parvec} \1_{}} 
= \frac{\ivd{m} \retrokmod_{m, n} \addf[n]{k}}{\ivd{m} \retrokmod_{m, n} \1_{}} - \frac{\ivd{m - 1} \retrokmod_{m - 1, n} \addf[n]{k}}{\ivd{m - 1} \retrokmod_{m - 1, n} \1_{}}. 
$$
Therefore,
\begin{multline}
\label{eq:error:decomp}
\vd{0:n} \addf[n]{k} - \post{0:n}^\parvec \addf[n]{k} = 
\sum_{m = 1}^n \left( \frac{\ivd{m} \retrokmod_{m, n} \addf[n]{k}}{\ivd{m} \retrokmod_{m, n} \1_{}} - \frac{\ivd{m - 1} \retrokmod_{m - 1, n} \addf[n]{k}}{\ivd{m - 1} \retrokmod_{m - 1, n} \1_{}} \right) \\ + \frac{\ivd{0} \uk{0, n}{\parvec} \addf[n]{k}}{\ivd{0} \uk{0, n}{\parvec} \1_{}} - \frac{\chi^\parvec \md{0}{\parvec}\uk{0, n}{\parvec} \addf[n]{k}}{\chi^\parvec\md{0}{\parvec}\uk{0, n}{\parvec} \1_{}}.
\end{multline}
By Lemma~\ref{lem:initial:term}, 
$$
\left|\frac{\ivd{0} \uk{0, n}{\parvec} \addf[n]{k}}{\ivd{0} \uk{0, n}{\parvec} \1_{}} - \frac{\phi^\parvec_0\uk{0, n}{\parvec} \addf[n]{k}}{\phi^\parvec_0\uk{0, n}{\parvec} \1_{}}\right|\leqslant 2c_0(\parvec,\parvar)\frac{\sigma_+}{\sigma_-}\|\addf{k}\|_\infty\eqsp.
$$

Consider now the error term at time $m>0$ in \eqref{eq:error:decomp}.  Define the kernel  
\begin{equation} \label{eq:def:retrokmod:joint}
    \tilde{\mathcal{L}}^{\parvar,\parvec}_{m, n}(x'_{m-1},x_m', \rmd x_{0:n}) \eqdef \delta_{x'_{m-1}}(\rmd x_{m-1})  \prod_{\ell = 0}^{m - 2} \kvd{\ell \vert \ell+1}(x_{\ell + 1}, \rmd x_\ell) \delta_{x_m'}(\rmd x_m) \prod_{\ell = m}^{n - 1} \uk{\ell}{\parvec}(x_\ell, \rmd x_{\ell + 1}), 
\end{equation}
on  $(\mathbb{R}^d)^2 \times \mathcal{B}((\mathbb{R}^d)^{n+1})$ so that for all $x_{m-1}$, $x_m\in\mathbb{R}^d$,
$$
\tilde{\mathcal{L}}^{\parvar,\parvec}_{m, n}\addf[n]{k}(x_{m-1},x_m) = \left\{
    \begin{array}{ll}
      \kvd{m-2 \vert m-1}\ldots \kvd{k \vert k+1}\addf{k}(x_{m-1})\uk{m, n}{\parvec}\1(x_m) & \mbox{if } k\leq m-2\,, \\
      \addf{k}(x_{m-1},x_m)\uk{m, n}{\parvec}\1(x_m) & \mbox{if } k = m-1\,, \\
       \uk{m, n}{\parvec}\addf{k}(x_m)  & \mbox{if } k\geq m\,.
    \end{array}
\right.
$$
Then, write
$$
\frac{\ivd{m} \retrokmod_{m, n} \addf[n]{k}}{\ivd{m} \retrokmod_{m, n} \1_{}}
- \frac{\ivd{m - 1} \retrokmod_{m - 1, n} \addf[n]{k}}{\ivd{m - 1} \retrokmod_{m - 1, n} \1_{}} 
=
\frac{\ivd{m} \kvd{m-1 \vert m} \tilde{\mathcal{L}}^{\parvar,\parvec}_{m, n} \addf[n]{k}}{\ivd{m} \retrokmod_{m, n} \1_{}}
- \frac{\ivd{m - 1} \uk{m-1}{\parvec} \tilde{\mathcal{L}}^{\parvar,\parvec}_{m, n} \addf[n]{k}}{\ivd{m - 1} \retrokmod_{m - 1, n} \1_{}} \eqsp.
$$
Let $1\leq m \leq n$ and $x_{m-1}^*$ and $x_m^*$ be arbitrary elements in $\mathbb{R}^d$. For $k\neq m-1$, define 
\begin{align}
\mathcal{L}^{*,\parvar,\parvec}_{m, n}\addf[n]{k}(x_{m-1},x_m) &= \frac{\tilde{\mathcal{L}}^{\parvar,\parvec}_{m, n}\addf[n]{k}(x_{m-1},x_m) }{\tilde{\mathcal{L}}^{\parvar,\parvec}_{m, n}\1(x_{m-1},x_m)} - \frac{\tilde{\mathcal{L}}^{\parvar,\parvec}_{m, n}\addf[n]{k}(x^*_{m-1},x^*_m) }{\tilde{\mathcal{L}}^{\parvar,\parvec}_{m, n}\1(x^*_{m-1},x^*_m)}\eqsp,\label{eq:def:lstar}\\
&= \frac{\tilde{\mathcal{L}}^{\parvar,\parvec}_{m, n}\addf[n]{k}(x_{m-1},x_m) }{\uk{m, n}{\parvec}\1(x_m)} - \frac{\tilde{\mathcal{L}}^{\parvar,\parvec}_{m, n}\addf[n]{k}(x^*_{m-1},x^*_m) }{\uk{m, n}{\parvec}\1(x^*_m)}\nonumber
\end{align}
and for $k=m-1$, $\mathcal{L}^{*,\parvar,\parvec}_{m, n}\addf[n]{k}(x_{m-1},x_m) = \addf{k}(x_{m-1},x_m)$. By Lemma~\ref{lem:geo:bound}, $\left\|\mathcal{L}^{*,\parvar,\parvec}_{m, n} \addf[n]{k}\right\|_\infty$ can be upper bounded and note that
\begin{multline*}
 \frac{\ivd{m} \retrokmod_{m, n} \addf[n]{k}}{\ivd{m} \retrokmod_{m, n} \1_{}} - \frac{\ivd{m - 1} \retrokmod_{m - 1, n} \addf[n]{k}}{\ivd{m - 1} \retrokmod_{m - 1, n} \1_{}}  \\= \frac{\ivd{m} \kvd{m-1 \vert m} \left\{\mathcal{L}^{*,\parvar,\parvec}_{m, n} \addf[n]{k}\tilde{\mathcal{L}}^{\parvar,\parvec}_{m, n} \1\right\}}{\ivd{m} \retrokmod_{m, n} \1_{}} - \frac{\ivd{m - 1} \uk{m-1}{\parvec} \left\{\mathcal{L}^{*,\parvar,\parvec}_{m, n} \addf[n]{k}\tilde{\mathcal{L}}^{\parvar,\parvec}_{m, n} \1\right\}}{\ivd{m - 1} \retrokmod_{m - 1, n} \1_{}}\eqsp. 
\end{multline*}
Define now the normalized measure $\tilde\phi^\parvar_mh$ by $\ivd{m - 1} \uk{m-1}{\parvec}h/\ivd{m - 1} \uk{m-1}{\parvec}\1$, so that
\begin{align*}
 \frac{\ivd{m} \retrokmod_{m, n} \addf[n]{k}}{\ivd{m} \retrokmod_{m, n} \1_{}} &- \frac{\ivd{m - 1} \retrokmod_{m - 1, n} \addf[n]{k}}{\ivd{m - 1} \retrokmod_{m - 1, n} \1_{}}  \\
&=  \frac{\ivd{m} \kvd{m-1 \vert m} \left\{\mathcal{L}^{*,\parvar,\parvec}_{m, n} \addf[n]{k}\tilde{\mathcal{L}}^{\parvar,\parvec}_{m, n} \1\right\}}{\ivd{m} \retrokmod_{m, n} \1_{}} - \frac{\tilde\phi^\parvar_m\left\{\mathcal{L}^{*,\parvar,\parvec}_{m, n} \addf[n]{k}\tilde{\mathcal{L}}^{\parvar,\parvec}_{m, n} \1\right\}}{\tilde\phi^\parvar_m\tilde{\mathcal{L}}^{\parvar,\parvec}_{m, n}\1_{}} \\ \\
&= \frac{\ivd{m} \kvd{m-1 \vert m} \left\{\mathcal{L}^{*,\parvar,\parvec}_{m, n} \addf[n]{k}\tilde{\mathcal{L}}^{\parvar,\parvec}_{m, n} \1\right\} - \tilde\phi^\parvar_m \left\{\mathcal{L}^{*,\parvar,\parvec}_{m, n} \addf[n]{k}\tilde{\mathcal{L}}^{\parvar,\parvec}_{m, n} \1\right\}}{\tilde\phi^\parvar_m\tilde{\mathcal{L}}^{\parvar,\parvec}_{m, n}\1_{}} \\
&\hspace{2.5cm}+  \frac{\ivd{m} \kvd{m-1 \vert m}\left\{\mathcal{L}^{*,\parvar,\parvec}_{m, n} \addf[n]{k}\tilde{\mathcal{L}}^{\parvar,\parvec}_{m, n} \1\right\}}{\ivd{m} \retrokmod_{m, n} \1_{}}\left(\frac{\tilde\phi^\parvar_m\tilde{\mathcal{L}}^{\parvar,\parvec}_{m, n}\1-\ivd{m} \retrokmod_{m, n} \1_{}}{\tilde\phi^\parvar_m\tilde{\mathcal{L}}^{\parvar,\parvec}_{m, n}\1_{}}\right)\eqsp.
\end{align*}
Then, using that
$$
\left|\frac{\ivd{m} \kvd{m-1 \vert m} \left\{\mathcal{L}^{*,\parvar,\parvec}_{m, n} \addf[n]{k}\tilde{\mathcal{L}}^{\parvar,\parvec}_{m, n} \1\right\}}{\ivd{m} \retrokmod_{m, n} \1_{}}\right|\leqslant \left\|\mathcal{L}^{*,\parvar,\parvec}_{m, n} \addf[n]{k}\right\|_\infty\eqsp,
$$
and by H\ref{assum:bias:bound},
\begin{align*}
\left|\frac{\tilde\phi^\parvar_m\tilde{\mathcal{L}}^{\parvar,\parvec}_{m, n} \1_{}-\ivd{m} \retrokmod_{m, n} \1_{}}{\tilde\phi^\parvar_m\tilde{\mathcal{L}}^{\parvar,\parvec}_{m, n} \1_{}}\right|&\leqslant c_m(\parvec,\parvar)\frac{\left\|\tilde{\mathcal{L}}^{\parvar,\parvec}_{m, n} \1\right\|_\infty}{\tilde\phi^\parvar_m\tilde{\mathcal{L}}^{\parvar,\parvec}_{m, n}\1_{}}\eqsp,\\
\left|\frac{\ivd{m} \kvd{m-1 \vert m} \left\{\mathcal{L}^{*,\parvar,\parvec}_{m, n} \addf[n]{k}\tilde{\mathcal{L}}^{\parvar,\parvec}_{m, n} \1\right\} - \tilde\phi^\parvar_m\left\{\mathcal{L}^{*,\parvar,\parvec}_{m, n} \addf[n]{k}\tilde{\mathcal{L}}^{\parvar,\parvec}_{m, n} \1\right\}}{\tilde\phi^\parvar_m\tilde{\mathcal{L}}^{\parvar,\parvec}_{m, n} \1_{}}\right| &\leqslant c_m(\parvec,\parvar)\frac{\left\|\mathcal{L}^{*,\parvar,\parvec}_{m, n} \addf[n]{k}\right\|_\infty\left\|\tilde{\mathcal{L}}^{\parvar,\parvec}_{m, n} \1\right\|_\infty}{\tilde\phi^\parvar_m\tilde{\mathcal{L}}^{\parvar,\parvec}_{m, n}\1_{}}\eqsp,
\end{align*}
yields
$$
\left| \frac{\ivd{m} \retrokmod_{m, n} \addf[n]{k}}{\ivd{m} \retrokmod_{m, n} \1_{}} - \frac{\ivd{m - 1} \retrokmod_{m - 1, n} \addf[n]{k}}{\ivd{m - 1} \retrokmod_{m - 1, n} \1_{}} \right|\leqslant 2c_m(\parvec,\parvar)\frac{\left\|\mathcal{L}^{*,\parvar,\parvec}_{m, n} \addf[n]{k}\right\|_\infty\left\|\tilde{\mathcal{L}}^{\parvar,\parvec}_{m, n} \1\right\|_\infty}{\tilde\phi^\parvar_m\tilde{\mathcal{L}}^{\parvar,\parvec}_{m, n}\1_{}}\eqsp.
$$
Note also that by H\ref{assum:strong:mixing},
$$
\tilde\phi^\parvar_m\tilde{\mathcal{L}}^{\parvar,\parvec}_{m, n}\1_{} \geqslant \sigma_- \mu  \uk{m+1, n-1}{\parvec}\1\eqsp,
$$ 
and for all $x_m\in\mathbb{R}^d$,
$$
\tilde{\mathcal{L}}^{\parvar,\parvec}_{m, n} \1(x_m)\leqslant  \sigma_+\mu  \uk{m+1, n-1}{\parvec}\1\eqsp.
$$
Therefore,
$$
\left| \frac{\ivd{m} \retrokmod_{m, n} \addf[n]{k}}{\ivd{m} \retrokmod_{m, n} \1_{}} - \frac{\ivd{m - 1} \retrokmod_{m - 1, n} \addf[n]{k}}{\ivd{m - 1} \retrokmod_{m - 1, n} \1_{}} \right|\leqslant 2\frac{\sigma_+}{\sigma_-}c_m(\parvec,\parvar) \left\|\mathcal{L}^{*,\parvar,\parvec}_{m, n} \addf[n]{k}\right\|_\infty\eqsp.
$$
The proof is completed using Lemma~\ref{lem:geo:bound}.

\section{Technical results}
\begin{lemma}
\label{lem:initial:term}
Assume that H\ref{assum:bias:bound} and  H\ref{assum:strong:mixing} hold. Then for all, $\parvec\in\parspace$, $\parvar\in\parvarspace$, $n\geq  1$, $k \in \intvect{0}{n - 1}$, bounded and measurable function $\addf{k}$, 
$$
\left|\frac{\ivd{0} \uk{0, n}{\parvec} \addf[n]{k}}{\ivd{0} \uk{0, n}{\parvec} \1_{}} - \frac{\chi^\parvec \md{0}{\parvec}\uk{0, n}{\parvec} \addf[n]{k}}{\chi^\parvec\md{0}{\parvec}\uk{0, n}{\parvec} \1_{}} \right|\leq 2c_0(\parvec,\parvar)\frac{\sigma_+}{\sigma_-}\|\addf{k}\|_\infty \eqsp,
$$
where $\addf[n]{k}$ is defined in \eqref{eq:def:addf}.
\end{lemma}

\begin{proof}
Consider the following decomposition of the first term:
\begin{align*}
\frac{\ivd{0} \uk{0, n}{\parvec} \addf[n]{k}}{\ivd{0} \uk{0, n}{\parvec} \1_{}} - \frac{\chi^\parvec \md{0}{\parvec}\uk{0, n}{\parvec} \addf[n]{k}}{\chi^\parvec\md{0}{\parvec}\uk{0, n}{\parvec} \1_{}} &= \frac{\ivd{0} \uk{0, n}{\parvec} \addf[n]{k}}{\ivd{0} \uk{0, n}{\parvec} \1_{}} - \frac{\phi^\parvec_0\uk{0, n}{\parvec} \addf[n]{k}}{\phi^\parvec_0\uk{0, n}{\parvec} \1_{}}\eqsp,\\
& = \frac{\ivd{0} \uk{0, n}{\parvec} \addf[n]{k}-\phi^\parvec_0\uk{0, n}{\parvec} \addf[n]{k}}{\ivd{0} \uk{0, n}{\parvec} \1_{}} \\
&\hspace{2cm}+  \frac{\phi^\parvec_0\uk{0, n}{\parvec} \addf[n]{k}}{\phi^\parvec_0\uk{0, n}{\parvec} \1_{}}\frac{\phi^\parvec_0\uk{0, n}{\parvec} \1_{} - \ivd{0} \uk{0, n}{\parvec} \1_{}}{\ivd{0} \uk{0, n}{\parvec} \1_{}} \eqsp,
\end{align*}
where $\phi^\parvec_0$ the filtering distribution at time $0$, i.e the law defined as  $\phi^\parvec_0h= \chi^\parvec \md{0}{\parvec}h/\chi^\parvec \md{0}{\parvec}$. Then, by H\ref{assum:bias:bound},
$$
\left|\frac{\ivd{0} \uk{0, n}{\parvec} \addf[n]{k}}{\ivd{0} \uk{0, n}{\parvec} \1_{}} - \frac{\phi^\parvec_0\uk{0, n}{\parvec} \addf[n]{k}}{\phi^\parvec_0\uk{0, n}{\parvec} \1_{}}\right|\leqslant 2c_0(\parvec,\parvar)\frac{\| \uk{0, n}{\parvec} \1_{}\|_\infty\|\addf[n]{k}\|_\infty}{\ivd{0} \uk{0, n}{\parvec} \1_{}}\eqsp.
$$
By H\ref{assum:strong:mixing}, for all $x_0\in \mathbb{R}^d$,
$$
\uk{0, n}{\parvec} \1_{}(x_0) = \int  \ud{0,\parvec}(x_0, x_{1})\mu(\rmd x_1)\uk{1, n}{\parvec} \1_{}(x_1)\leqslant \sigma_+  \int \mu(\rmd x_1)\uk{1, n}{\parvec} \1_{}(x_1)
$$
and 
$$
\ivd{0} \uk{0, n}{\parvec} \1 = \int \ivd{0} (\rmd x_0)\ud{0,\parvec}(x_0, x_{1})\mu(\rmd x_1)\uk{1, n}{\parvec} \1_{}(x_1)\geqslant \sigma_-  \int \mu(\rmd x_1)\uk{1, n}{\parvec} \1_{}(x_1)\eqsp,
$$
which yields
$$
\left|\frac{\ivd{0} \uk{0, n}{\parvec} \addf[n]{k}}{\ivd{0} \uk{0, n}{\parvec} \1_{}} - \frac{\phi^\parvec_0\uk{0, n}{\parvec} \addf[n]{k}}{\phi^\parvec_0\uk{0, n}{\parvec} \1_{}}\right|\leqslant 2c_0(\parvec,\parvar)\frac{\sigma_+}{\sigma_-}\|\addf{k}\|_\infty\eqsp.
$$
\end{proof}

\begin{lemma} \label{lem:geo:bound}
Assume that H\ref{assum:strong:mixing} holds. Then for all $n \in \nset$, $\parvec\in\parspace$, $\parvar\in\parvarspace$, $m \in \intvect{1}{n}$, $k \in \intvect{0}{n - 1}$, $x_{m-1},x_m, x^*_{m-1},x^*_m$ in $\mathbb{R}^d$, bounded and measurable function $\addf{k}$, 
$$
\left|\mathcal{L}^{*,\parvar,\parvec}_{m, n}\addf[n]{k}(x_{m-1},x_m) \right| \leq 
\left\{
    \begin{array}{ll}
        \| \addf{k} \|_\infty \rho^{m - k - 1}& \mbox{if } k\leq m-2\,, \\
        \| \addf{k} \|_\infty & \mbox{if } k = m-1\,, \\
         \| \addf{k} \|_\infty \rho^{k-m+1}& \mbox{if } k\geq m\,.
    \end{array}
\right.
$$
where $\rho = 1 - \sigma_-/\sigma_+$ and $\addf[n]{k}$ is defined in \eqref{eq:def:addf} and $\mathcal{L}^{*,\parvar,\parvec}_{m, n}\addf[n]{k}$ is defined in \eqref{eq:def:lstar}.
\end{lemma}

\begin{proof}
The proof is adapted from  \cite[Lemma~D.3]{gloaguen2022pseudo} and given here for completeness. Assume first that $k\leq m-2$. Then,
$$
\frac{\tilde{\mathcal{L}}^{\parvar,\parvec}_{m, n}\addf[n]{k}(x_{m-1},x_m) }{\uk{m, n}{\parvec}\1(x_m)} = \vd{m-2 \vert m-1}\ldots \kvd{k \vert k+1}\addf{k}(x_{m-1})
$$
Therefore,
$$
\frac{\tilde{\mathcal{L}}^{\parvar,\parvec}_{m, n}\addf[n]{k}(x_{m-1},x_m) }{\uk{m, n}{\parvec}\1(x_m)} - \frac{\tilde{\mathcal{L}}^{\parvar,\parvec}_{m, n}\addf[n]{k}(x^*_{m-1},x^*_m) }{\uk{m, n}{\parvec}\1(x^*_m)} =  (\delta_{x_{m-1}}-\delta_{x^*_{m-1}})\vd{m-2 \vert m-1} \ldots \kvd{k \vert k+1}\addf{k}\eqsp.
$$
By H\ref{assum:strong:mixing}, the Dobrushin coefficient of the variational backward kernels is upper-bounded by $1-\sigma_-/\sigma_+$ so that 
$$
\left|\frac{\tilde{\mathcal{L}}^{\parvar,\parvec}_{m, n}\addf[n]{k}(x_{m-1},x_m) }{\uk{m, n}{\parvec}\1(x_m)} - \frac{\tilde{\mathcal{L}}^{\parvar,\parvec}_{m, n}\addf[n]{k}(x^*_{m-1},x^*_m) }{\uk{m, n}{\parvec}\1(x^*_m)} \right|\leqslant  \left(1-\frac{\sigma_-}{\sigma_+}\right)^{m-k-1}\left\|\addf{k}\right\|_\infty\eqsp.
$$
In the case where $k = m-1$, 
$$
\frac{\tilde{\mathcal{L}}^{\parvar,\parvec}_{m, n}\addf[n]{k}(x_{m-1},x_m) }{\uk{m, n}{\parvec}\1(x_m)} = \addf{k}(x_{k},x_{k+1})\eqsp,
$$
so that the result is straightforward. 
Assume now first that $k\geq m$. Note that 
$$
\frac{\tilde{\mathcal{L}}^{\parvar,\parvec}_{m, n}\addf[n]{k}(x_{m-1},x_m) }{\uk{m, n}{\parvec}\1(x_m)} = \frac{\uk{m, n}{\parvec}\addf[n]{k}(x_{m-1},x_m) }{\uk{m, n}{\parvec}\1(x_m)} = \frac{F^\parvec_{m|n}\ldots F^\parvec_{k|n}\addf[n]{k}(x_m)\cdot\uk{m, n}{\parvec}\1_{}(x_m)}{\uk{m, n}{\parvec}\1_{}(x_m)}\eqsp,
$$
where the forward kernel $ \mathsf{F}^\parvec_{\ell|n}$ is given by
$$
 \mathsf{F}^\parvec_{\ell|n}h(x_\ell) = \frac{\uk{\ell}{\parvec}(h\uk{\ell+1, n-1}{\parvec}\1_{})(x_\ell)}{ \uk{\ell, n-1}{\parvec}\1_{}(x_\ell)}\eqsp.
$$
By H\ref{assum:strong:mixing},
$$
 \mathsf{F}^\parvec_{\ell|n}h(x_\ell) \geqslant \frac{\sigma_-}{\sigma_+}\mu_{\ell|n}h\eqsp,
$$
with $\mu_{\ell|n}h = \mu(h\uk{\ell+1, n-1}{\parvec}\1_{})(x_\ell) / \mu\uk{\ell+1, n-1}{\parvec}\1_{}$. Therefore, the Dobrushin coefficients of the kernels $F^\parvec_{\ell|n}$  are also upper-bounded by $1-\sigma_-/\sigma_+$. On the other hand,
$$
\frac{\tilde{\mathcal{L}}^{\parvar,\parvec}_{m, n}\addf[n]{k}(x_{m-1},x_m) }{\uk{m, n}{\parvec}\1(x_m)} - \frac{\tilde{\mathcal{L}}^{\parvar,\parvec}_{m, n}\addf[n]{k}(x^*_{m-1},x^*_m) }{\uk{m, n}{\parvec}\1(x^*_m)} = (\lambda_{m|n} - \lambda'_{m|n})\mathsf{F}^\parvec_{m|n}\ldots \mathsf{F}^\parvec_{k|n}\addf[n]{k}, 
$$
where $\lambda_{m|n}h = \delta_{x_m}h\uk{m, n}{\parvec}\1_{}/\delta_{x_m}\uk{m, n}{\parvec}\1_{}$ and $\lambda'_{m|n}h = \delta_{x'_m}h\uk{m, n}{\parvec}\1_{}/\delta_{x'_m}\uk{m, n}{\parvec}\1_{}$.
This yields
$$
\left|\frac{\tilde{\mathcal{L}}^{\parvar,\parvec}_{m, n}\addf[n]{k}(x_{m-1},x_m) }{\uk{m, n}{\parvec}\1(x_m)} - \frac{\tilde{\mathcal{L}}^{\parvar,\parvec}_{m, n}\addf[n]{k}(x^*_{m-1},x^*_m) }{\uk{m, n}{\parvec}\1(x^*_m)}\right|\leqslant  \left(1-\frac{\sigma_-}{\sigma_+}\right)^{k-m+1}\left\|\addf{k}\right\|_\infty\eqsp,
$$
which concludes the proof.
 
\end{proof}

\section{Deriving the recursive form of the ELBO}

To obtain a recursion on $T_n(X_n) = \pE_{\vd{0:n}}[\log p^{\parvec}_{0:n}(X_{0:n},Y_{0:n})/\vd{0:n}(X_{0:n}) \vert X_n]$, we notice, as in \cite{campbell2021online}, that $$
\vd{0:k}(x_{0:k}) = \vd{0:k-1}(x_{0:k-1})\bar{q}_{k\vert k-1}^\lambda(x_{k-1}, x_k) \eqsp,
$$ where $\bar{q}_{k\vert k-1}^\lambda(x_{k-1}, x_k) = \vd{k-1 \vert k}(x_k, x_{k-1})  \vd{k}(x_k) / \vd{k}(x_{k-1})$. The function $x_k \mapsto \bar{q}_{k\vert k-1}(x_{k-1}, x_k)$ is not the density of a Markov kernel but allows an alternate decomposition of the variational family forward in time. Since the density of the complete data model $x_{0:n}\mapsto p^{\parvec}_{0:n}(x_{0:n},Y_{0:n})$ also factorizes via the densities $x_k\mapsto \ell_k^\theta(x_{k-1}, x_k)$ of the forward kernels, the statistic $T_n(X_n)$ writes: 
$$
T_n(X_n) = \pE_{\vd{0:n}}\left[\log \frac{\chi^\theta(X_0) \md{0}{\theta}(X_0) \prod_{k=1}^n \ell_k^\theta(X_{k-1}, X_k)}{\vd{0}(X_0)\prod_{k=1}^n \bar{q}_{k\vert k-1}^\lambda(X_{k-1}, X_k)} \middle\vert X_n\right]\eqsp.
$$
By applying again the tower property of expectations, this yields:
\begin{align*}
T_n(X_n) &= \pE_{\vd{0:n}}\left[\pE_{\vd{0:n}}\left[\log \frac{\chi^\theta(X_0) \md{0}{\theta}(X_0) \prod_{k=1}^n \ell_k^\theta(X_{k-1}, X_k)}{\vd{0}(X_0)\prod_{k=1}^n \bar{q}_{k\vert k-1}^\lambda(X_{k-1}, X_k)} \middle\vert X_{n-1}, X_n \right] \middle\vert X_n\right] \\
&= \pE_{\vd{0:n}}\left[\pE_{\vd{0:n-1}}\left[\log \frac{\chi^\theta(X_0) \md{0}{\theta}(X_0) \prod_{k=1}^{n-1} \ell_k^\theta(X_{k-1}, X_k)}{\vd{0}(X_0)\prod_{k=1}^{n-1} \bar{q}_{k\vert k-1}^\lambda(X_{k-1}, X_k)} \middle\vert X_{n-1} \right]\right.\\
&\left. \hspace{7cm}+ \log \frac{\ell_n^\theta(X_{n-1}, X_n)}{\bar{q}_{n\vert n-1}^\lambda(X_{n-1}, X_n)}\middle\vert X_n\right]\eqsp.
\end{align*}
The inner expectation is $T_{k-1}(X_{k-1})$ by definition. Since all terms in the outer expectation are only functions of $X_{n-1}$, the expectation under $\vd{0:n}$ reduces to an expectation under the backward kernel $\vd{n-1 \vert n}$, i.e. 
$$
T_n(X_n) = \pE_{\vd{n-1 \vert n}} \left[T_{n-1}(X_{n-1})  + \log \frac{\ell_n^\theta(X_{n-1}, X_n)}{\bar{q}_{n\vert n-1}^\lambda(X_{n-1}, X_n)}\middle\vert X_n\right]\eqsp,
$$
which is the recursion proposed in (\ref{ref:eq:v_t}).

\section{Experiment details}
\subsection{Hardware configuration}

We ran all experiments on a machine with the following specifications.
\begin{itemize}
    \item CPUs: 4x Intel(R) Xeon(R) Gold 6154 (total 72 cores, 144 threads).
    \item RAM: 260 Go. 
\end{itemize}

No GPU was used. 
\subsection{Linear Gaussian models}
We provide here additional figures for the experiments of Section \ref{sec:experiments:linear_gaussian}. Figure \ref{fig:marginal_linear_gaussian_appendix} shows the marginal errors across time for the variational models for the different stopping points of Figure~\ref{fig:linear_gaussian}. Table \ref{fig:table_linear_gaussian_appendix} shows the accuracy of the optimal Kalman smoothing (with true parameters $\theta$) w.r.t the true states, as well as the numerical values for the smoothing errors at the three stopping points of the optimization.

We also provide examples of smoothed states for the multivariate case. In Figure \ref{fig:smoothing_multidimensional}, we plot the paths of an evaluation sequence where the state space is of dimension 3 and the observation space is of dimension 4. We visualise the results by marginalizing $\post{0:n}^\parvec$ $\vd{0:n}$ on each dimension of the state space (and for each timestep). Note that here we do not learn the emission matrix and the variances to avoid having to fix indeterminacies of the multidimensional case (scaling and permutations).

\begin{figure}
  \centering
  \includegraphics[width=\linewidth, height=0.6\textwidth]{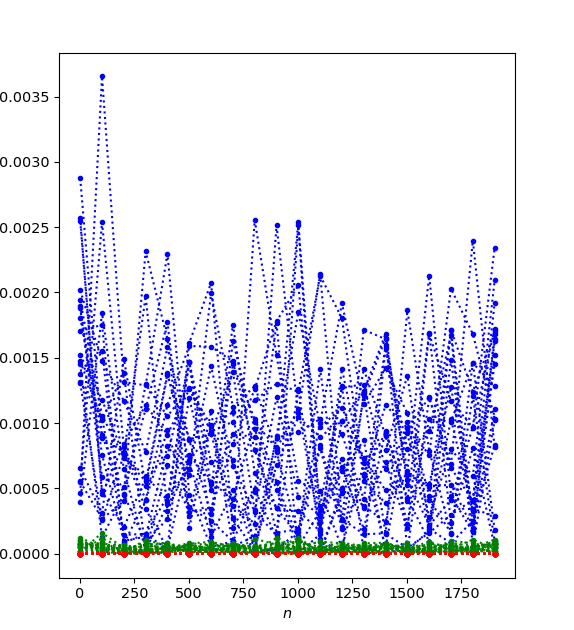}
  \caption{Marginal errors between the variational model and the true model at three different points of the optimization: 60 epochs (blue), 80 epochs (green) and 100 epochs (red).}
  \label{fig:marginal_linear_gaussian_appendix}
\end{figure}%

\begin{figure}
    \centering
    \small
    \begin{tabular}{cccccc}
    \toprule
    {Seq nb.} & $\mathrm{{Mean~err._{\theta}}}$ & $\mathrm{{Var~err._{\theta}}}$ & $\mathrm{Smoothing~err._{{\theta} / {\lambda_{60}}}}$ & $\mathrm{Smoothing~err._{{\theta} / {\lambda_{80}}}}$ & $\mathrm{Smoothing~err._{{\theta} / {\lambda_{100}}}}$ \\
    \midrule
    0  & 0.001432 & 0.000004 & 0.954874 & 0.054593 & 0.000232 \\
    1  & 0.001432 & 0.000004 & 0.914786 & 0.052843 & 0.000227 \\
    2  & 0.001466 & 0.000004 & 1.039261 & 0.058327 & 0.000243 \\
    3  & 0.001420 & 0.000004 & 0.953855 & 0.054565 & 0.000232 \\
    4  & 0.001394 & 0.000004 & 0.943506 & 0.054108 & 0.000231 \\
    5  & 0.001473 & 0.000004 & 1.088147 & 0.060461 & 0.000249 \\
    6  & 0.001445 & 0.000004 & 0.944956 & 0.054160 & 0.000231 \\
    7  & 0.001415 & 0.000004 & 0.994196 & 0.056330 & 0.000237 \\
    8  & 0.001390 & 0.000004 & 0.849939 & 0.050003 & 0.000219 \\
    9  & 0.001404 & 0.000004 & 0.990084 & 0.056150 & 0.000237 \\
    10 & 0.001408 & 0.000004 & 0.959414 & 0.054805 & 0.000233 \\
    11 & 0.001378 & 0.000004 & 0.883776 & 0.051488 & 0.000223 \\
    12 & 0.001412 & 0.000004 & 0.952297 & 0.054494 & 0.000232 \\
    13 & 0.001405 & 0.000004 & 0.793487 & 0.047527 & 0.000212 \\
    14 & 0.001497 & 0.000005 & 0.996393 & 0.056432 & 0.000237 \\
    15 & 0.001495 & 0.000004 & 0.893889 & 0.051933 & 0.000225 \\
    16 & 0.001564 & 0.000005 & 0.916225 & 0.052907 & 0.000227 \\
    17 & 0.001406 & 0.000004 & 0.794842 & 0.047594 & 0.000212 \\
    18 & 0.001386 & 0.000004 & 0.968786 & 0.055223 & 0.000234 \\
    19 & 0.001373 & 0.000004 & 0.970476 & 0.055282 & 0.000234 \\
    \bottomrule
    \end{tabular}
    \caption{First column: empirical mean of $\{(\hat{x}_{k, \theta} - x_k^{*})^2\}_{0 \leqslant k \leqslant n}$ where $x_k^{*}$ is the true state and $\hat{x}_{k, \theta}$ is the marginal mean of $\post{0:n}^\parvec$ at time $k$ provided by Kalman smoothing with true parameters $\theta$. Second column: empirical variance of the same quantity. Third, fourth and fifth columns: smoothing errors $\big| \vd{0:n} h_{0:n} -  \post{0:n}^{\parvec} h_{0:n} \big|$ for $\tilde{h}_{k}(x_k, x_{k+1}) = x_k$ at $n = 2000$, when $\post{0:n}^\parvec$ is given via Kalman smoothing with the true parameters $\theta$ and $\vd{0:n}$ is given via Kalman smoothing with parameters $\lambda$ selected at epochs $60$,$80$ and $100$. Each line is corresponds to one observation sequence in $(Y^j_{0:n})_{1\leq j \leq J}$, $J=20$.}
    \label{fig:table_linear_gaussian_appendix}
\end{figure}

\begin{figure}
  \centering
  \includegraphics[width=\linewidth, height=0.8\textwidth]{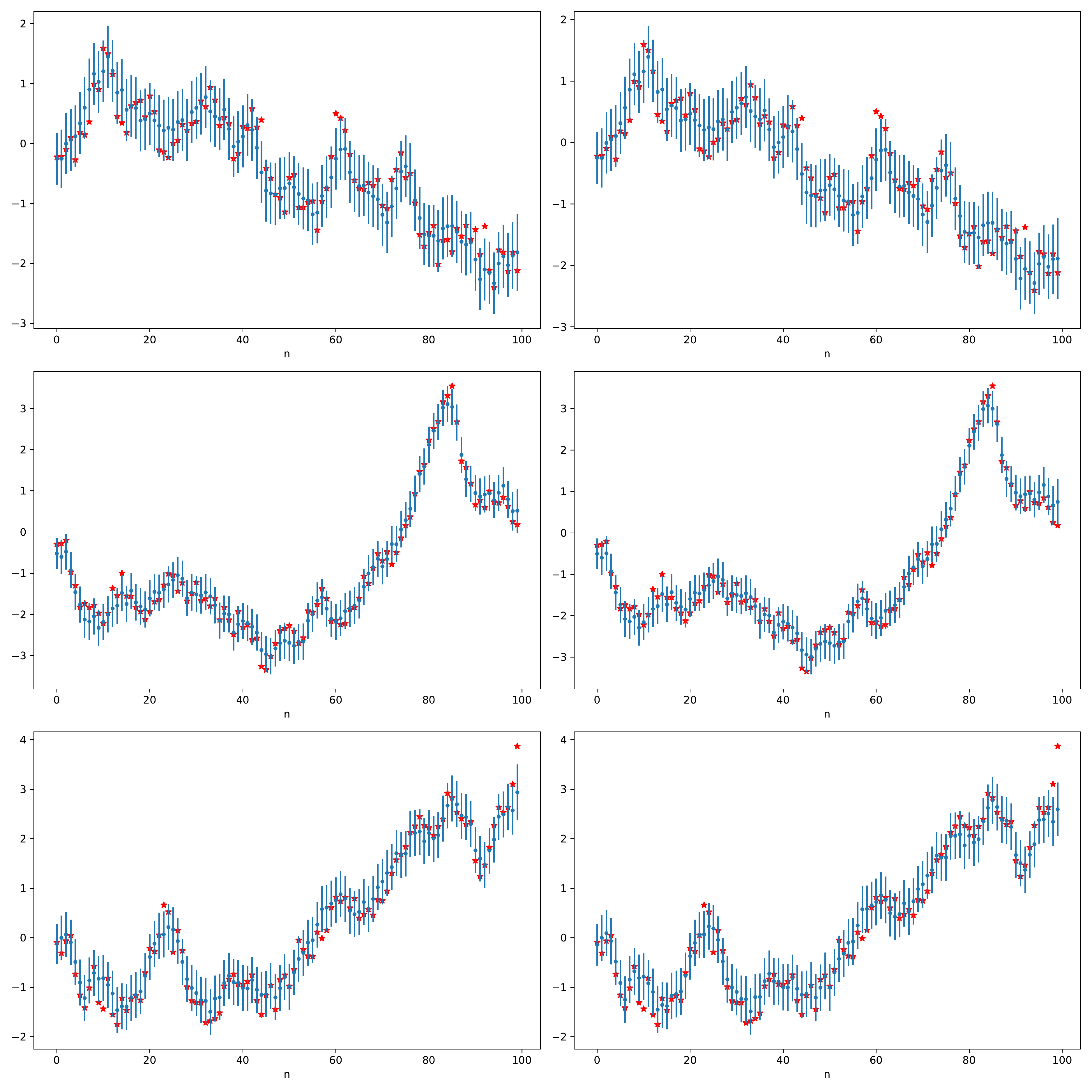}
  \caption{Example of smoothed states when the dimension of the state space is 3 and the dimension of the observations is 4. Left column: component-wise (from top to bottom) smoothed states with true parameters $\theta$. Right column: same thing with learnt parameters $\lambda$. Red stars: true state components. Blue dots: smoothed marginal means of each component. Blue vertical lines: 95\% confidence regions built from the smoothed marginal variances of each component. The horizontal axis is the time axis.}
  \label{fig:smoothing_multidimensional}
\end{figure}%

\subsection{Nonlinear models}

Here we provide additional details on the experiments of section \ref{sec:experiments:nonlinear}. 

\begin{itemize}
    \item For the nonlinear emission function $h^\theta$ of the data model, we used a single-layer perceptron with a $\tanh$ activation function. We found that this mapping is sufficiently nonlinear to evaluate and compare the models, but we further apply the $\cos$ function to the output to ensure noninjectivity. 
    
    \item For the parameterization of the method based on \cite{johnson2016}, the mapping $f_{enc}^\lambda$ is a multi-layer perceptron (MLP) with two hidden layers of 16 neurons and a $\tanh$ activation function. The activation function is not applied to the output layer to ensure that the values can exceed values outside the range $[-1,1]$, being natural parameters of Gaussian distributions. The output of the network is split into two natural parameters $\eta_1$ and $\eta_2$, the latter being constrained to strictly negative values by applying the softplus function $x \mapsto -\log(1 +e^x)$. We use Xavier initialization for the matrix parameters, and random normal initialisation for the bias parameters. 
    \item For the parameterization of $r^\lambda$ in our method, we use the exact same MLP as $f_{enc}^\lambda$ described above but take the predictive parameters $u_k$ as additinal input (we denote it $f^{'\lambda}_{enc}$). We add a forget gate to mitigate vanishing / exploding gradient issues, where the forget state is computed by a single-layer perceptron. If we denote by $s$ this forget layer, then $$r^\lambda(u_k,y_k) = s(u_k,y_k) * u_k + \left[1 - s(u_k, y_k)\right] * f^{'\lambda}_{enc}(u_k, y_k)\eqsp,$$ where $*$ is the element-wise product.
\end{itemize}

\end{document}